\documentclass[a4paper,12pt]{article}
\usepackage[utf8]{inputenc}
\usepackage[T1]{fontenc}
\usepackage[american]{babel}
\usepackage{xspace}
\usepackage{hyperref}
\usepackage{amssymb}
\usepackage{amsthm}
\usepackage{amsmath}
\usepackage{mathrsfs}
\usepackage{tikz}
\usepackage{boxedminipage,verbatim,float,caption}
\usepackage{enumerate}
\usepackage{framed}
\usepackage{multicol}
\usepackage{pdfpages}
\usepackage{fullpage}
\usepackage{lineno}

\usepackage[affil-it]{authblk}
\usepackage{etoolbox}
\makeatletter
\patchcmd{\@maketitle}{\LARGE}{\fontsize{16}{17}\selectfont}{}{}
\makeatother

\providecommand{\keywords}[1]
{
	\small	
	{\textit{Keywords:}} #1
}

\def\bN{\mathbb{N}}
\def\bQ{\mathbb{Q}}

\def\bI{\mathbb{I}}

\def\cA{\mathcal{A}}
\def\cB{\mathcal{B}}
\def\cC{\mathcal{C}}

\def\cL{\mathcal{L}}

\def\cP{\mathcal{P}}

\def\cS{\mathcal{S}}
\def\cX{\mathcal{X}}
\def\cY{\mathcal{Y}}

\newcommand{\VC}{\mathsf{VC}}
\newcommand{\vc}{\mathsf{vc}}

\renewcommand{\operatorname}[1]{\mathsf{#1}}

\newcommand{\FPT}{\textsf{FPT}\xspace}
\newcommand{\XP}{\textsf{XP}\xspace}
\newcommand{\W}{\textsf{W}}

\newcommand{\NP}{\textsf{NP}\xspace}
\newcommand{\Poly}{\textsf{P}\xspace}
\newcommand{\APX}{\textsf{APX}\xspace}

\def\wrt{{w.r.t.}\xspace}
\def\ie{{i.e.}\xspace}

\newcommand{\nec}{\mathsf{nec}}
\newcommand{\snec}{\mathsf{s\text{-}nec}}
\newcommand{\mim}{\mathsf{mim}}
\newcommand{\rw}{\mathsf{rw}}
\newcommand{\Qrw}{\mathsf{rw}_\bQ}

\newcommand{\mw}{\mathsf{mw}}

\newcommand{\f}{\mathsf{f}}

\newtheorem{theorem}{Theorem}
\newtheorem{lemma}[theorem]{Lemma}

\newtheorem{corollary}[theorem]{Corollary}

\newtheorem{claim}[theorem]{Claim}
\newtheorem{fact}[theorem]{Fact}
\newtheorem{definition}[theorem]{Definition}

\newtheorem{observation}[theorem]{Observation}

\newcommand{\cc}{\operatorname{cc}}
\newcommand{\comp}[1]{\overline{#1}}

\newcommand{\w}{\mathsf{w}}
\renewcommand{\min}{\operatorname{min}}
\renewcommand{\max}{\operatorname{max}}

\newcommand{\pbDef}[3]{
	\noindent
	\vspace{0.08cm}
	\begin{center}
		\begin{boxedminipage}{0.98 \columnwidth}
			\textsc{#1}\\[5pt]
			\textbf{Input:}  #2.\\
			\textbf{Output:}  #3
		\end{boxedminipage}
	\end{center}
	\vspace{0.08cm}	}

\renewcommand{\emptyset}{\varnothing}

\newcommand{\reduce}{\operatorname{reduce}}
\newcommand{\best}{\operatorname{best}}
\def\Rep#1#2{\mathcal{R}_{#2}^{#1}}
\def\rep#1#2{\mathsf{rep}_{#2}^{#1}}
\def\equi#1#2{\equiv_{#2}^{#1}}
\renewcommand{\leq}{\leqslant}

\newcommand{\aux}{\operatorname{aux}}

\usepackage{todonotes}
\setuptodonotes{color=green!40}

\usepackage[mathscr]{euscript}
\def\cX{\mathscr{X}}
\def\cY{\mathscr{Y}}

\begin{document}
	
	\title{Node Multiway Cut and Subset Feedback Vertex Set on Graphs of Bounded Mim-width}

	\author[1,*]{Benjamin Bergougnoux}
	\author[2]{Charis Papadopoulos}
	\author[1]{Jan Arne Telle}
	\affil[1]{University of Bergen, Norway}
	\affil[2]{University of Ioannina, Greece}
	\affil[*]{Corresponding author: benjamin.bergougnoux@uib.no}
	
	\maketitle
	
	\thanks{An extended abstract appeared in the proceedings of	46th International Workshop on Graph-Theoretic Concepts in Computer Science, WG 2020.}
	\begin{abstract}
		The two weighted graph problems \textsc{Node Multiway Cut} (NMC) and \textsc{Subset Feedback Vertex Set} (SFVS)
		both ask for a vertex set of minimum total weight, that for NMC
		disconnects a given set of terminals, and for SFVS intersects all cycles containing a vertex of a given set.
		We design a meta-algorithm that allows to solve both problems in time $2^{O(rw^3)}\cdot n^{4}$,  $2^{O(q^2\log(q))}\cdot n^{4}$, and $n^{O(k^2)}$ where $rw$ is the rank-width, $q$ the $\bQ$-rank-width, and $k$ the mim-width of a given decomposition.
		This answers in the affirmative an open question raised by Jaffke et~al.~(Algorithmica, 2019) concerning an \XP algorithm for SFVS parameterized by mim-width.
		
		By a unified algorithm, this solves both problems in polynomial-time on the following graph classes:
		\textsc{Interval}, \textsc{Permutation}, and \textsc{Bi-Interval} graphs,
		\textsc{Circular Arc} and \textsc{Circular
			Permutation} graphs, \textsc{Convex} graphs, \textsc{$k$-Polygon},
		\textsc{Dilworth-$k$} and \textsc{Co-$k$-Degenerate} graphs for fixed
		$k$; and also on \textsc{Leaf Power} graphs if a leaf root is given as input,
		on \textsc{$H$-Graphs} for fixed $H$ if an $H$-representation is given as input,
		and on arbitrary powers of graphs in all the above classes.
		Prior to our results, only SFVS was known to be tractable restricted only on \textsc{Interval} and \textsc{Permutation} graphs, whereas all other results are new.
	\end{abstract}

		
	\section{Introduction }	
		Given a vertex-weighted graph $G$ and a set $S$ of its vertices, the \textsc{Subset Feedback Vertex Set} (SFVS) problem asks for a vertex set of minimum weight
	that intersects all cycles containing a vertex of $S$.
	SFVS was introduced by Even et~al.~\cite{EvenNZ00} who proposed an $8$-approximation algorithm.
	Cygan et~al.~\cite{CyganPPW13} and Kawarabayashi and Kobayashi \cite{KK12} independently
	showed that SFVS is fixed-parameter tractable (\FPT) parameterized by the solution size, while Hols and Kratsch \cite{HolsK18} provide a randomized polynomial kernel for the problem.
	As a generalization of the classical \NP-complete \textsc{Feedback Vertex Set} (FVS) problem, for which $S=V(G)$, there has been a considerable amount of work to obtain faster algorithms for SFVS, both for general graphs~\cite{ChitnisFLMRS17,FominGLS16} where the current best is an $O^*(1.864^{n})$ algorithm due to Fomin et~al.~\cite{FominHKPV14}, and restricted to special graph classes~\cite{BergougnouxBBK20,GolovachHKS14,PapadopoulosT20,PapadopoulosT19}.
	Naturally, FVS and SFVS differ in complexity, as exemplified by
	split graphs where FVS is polynomial-time solvable \cite{fvs:chord:corneil:1988} whereas SFVS remains \NP-hard \cite{FominHKPV14}.
	Moreover, note that the vertex-weighted variation of SFVS behaves differently than the unweighted one, as exposed on graphs with bounded independent set sizes: weighted SFVS is \NP-complete on graphs with independent set size at most four, whereas unweighted SFVS is in \XP parameterized by the independent set size \cite{PapadopoulosT20}.
	
	Closely related to SFVS is the \NP-hard \textsc{Node Multiway Cut} (NMC) problem in which we are given a
	vertex-weighted graph $G$ and a set $T$ of (terminal) vertices, and asked to find a vertex set of minimum weight that disconnects all the terminals \cite{Calinescu08,GargVY04}. 
	NMC is a well-studied problem in terms of approximation \cite{GargVY04}, as well as parameterized algorithms \cite{Calinescu08,CLL09,ChitnisFLMRS17,CPPW13,DahlhausJPSY94,FominHKPV14}.
	It is not difficult to see that SFVS for $S = \{v\}$ coincides with NMC in which $T=N(v)$.
	In fact, NMC reduces to SFVS by adding a single vertex $v$ with a large weight that is adjacent to all terminals 
    and setting $S=\{v\}$ \cite{FominHKPV14}.
	Thus, in order to solve NMC on a given graph one may apply a known algorithm for SFVS on a vertex-extended graph.
	Observe, however, that through such an approach one needs to clarify that the vertex-extended graph still obeys the necessary properties of the known algorithm for SFVS.
	This explains why most of the positive results on SFVS on graph families \cite{GolovachHKS14,PapadopoulosT19,PapadopoulosT20} can not be translated to NMC.
	
	In this paper, we investigate the complexity of SFVS and NMC when parameterized by structural graph width parameters.
	Well-known graph width parameters include tree-width \cite{Bodlaender06}, clique-width \cite{CourcelleO00}, rank-width \cite{Oum05a}, and maximum induced matching width (\emph{a.k.a.} mim-width) \cite{Vatshelle12}.
	These are of varying strength, with tree-width of modeling power strictly weaker than clique-width, as it is bounded on a proper subset of the graph classes having bounded clique-width, with rank-width and clique-width of the same modeling power, and with mim-width much stronger than clique-width. Belmonte and Vatshelle~\cite{BelmonteV13} showed that several graph classes, like interval graphs and permutation graphs, have bounded mim-width and a decomposition witnessing this can be found in polynomial time, whereas it is known that the clique-width of such graphs can be proportional to the square root of the number of vertices \cite{GolumbicR00}.
	In this way, an \XP algorithm parameterized by mim-width has the feature of unifying several algorithms on well-known graph classes.
	
	We obtain most of these parameters through the well-known notion of \emph{branch-decomposition} introduced in \cite{RobertsonS91}.
	This is a natural hierarchical clustering of $G$, represented as a subcubic tree $T$ with the vertices of $G$ at its leaves.
	Any edge of the tree defines a cut of $G$ given by the leaves of the two subtrees that result from removing the edge from $T$.
	Judiciously choosing a cut-function to measure the complexity of such cuts, or rather of the bipartite subgraphs of $G$ given by the edges crossing the cuts, this framework then defines a graph width parameter by a minmax relation, minimum over all trees and maximum over all its cuts.
	Several graph width parameters have been defined this way, like carving-width, maximum matching-width, boolean-width etc.
	We will in this paper focus on: (i) rank-width~\cite{Oum05a} whose cut function is the $GF[2]$-rank of the adjacency matrix, (ii) $\bQ$-rank-width~\cite{OumSV13} a variant of rank-width with interesting algorithmic properties which instead uses the rank over the rational field, and (iii) mim-width~\cite{Vatshelle12} whose cut function is the size of a maximum induced matching of the graph crossing the cut.
	Concerning their computations, for rank-width and $\bQ$-rank-width, there are $2^{3k}\cdot n^{4}$ time algorithms that, given a graph $G$ as input and $k\in \bN$, either output a decomposition for $G$ of width at most $3k+1$ or confirms that the width of $G$ is more than $k$ \cite{OumSV13,OumS06}. However, it is not known whether the mim-width of a graph can be approximated within a constant factor in time $n^{f(k)}$ for some function $f$.
	
	\medskip
	Let us mention what is known regarding the complexity of NMC and SFVS parameterized by these width measures.
	Since these problems can be expressed in MSO$_1$-logic it follows that
	they are \FPT parameterized by tree-width, clique-width, rank-width or $\bQ$-rank-width \cite{CourcelleMR00,Oum09a}, however the runtime will contain a tower of 2's with more than $4$ levels.
	Recently, Bergougnoux et al.~\cite{BergougnouxBBK20} proposed $k^{O(k)}\cdot n^{3}$ time algorithms for these two problems parameterized by treewidth and proved that they cannot be solved in time $k^{o(k)}\cdot n^{O(1)}$ unless ETH fails.
	For mim-width, we know that FVS and thus SFVS are both \W[1]-hard when parameterized by the mim-width of a given decomposition \cite{JaffkeKT20}.
	
	Attacking SFVS seems to be a hard task that requires more tools than for FVS.
	Even for very small values of mim-width that capture several graph classes, the tractability of SFVS, prior to our result, was left open besides interval and permutation graphs \cite{PapadopoulosT19}. Although FVS was known to be tractable on such graphs for more than a decade \cite{KratschMT08}, the complexity status of SFVS still remained unknown.
	
	\medskip
	\paragraph{\bf Our results.}
	We design a meta-algorithm that, given a graph and a branch-decomposition, solves SFVS (or NMC via its reduction to SFVS).
	The runtime of this algorithm is upper bounded by  $2^{O(rw^3)}\cdot n^4$, $2^{O(q^2\log(q))}\cdot n^4$ and $n^{O(k^2)}$ where $rw$, $q$ and $k$ are the rank-width, the $\bQ$-rank-width and the mim-width of the given branch-decomposition.
	For clique-width, our meta-algorithm implies that we can solve SFVS and NMC in time $2^{O(k^2)}\cdot n^{O(1)}$ where $k$ is the clique-width of a given clique-width expression.
	However, we do not prove this as it is not asymptotically optimal, indeed Jacob et al.~\cite{JacobBDP21} show recently that SFVS and NMC is solvable in time $2^{O(k\log k )}\cdot n$ given a clique-width expression.
	
	We resolve in the affirmative the question raised by Jaffke et~al.~\cite{JaffkeKT20}, also mentioned in \cite{PapadopoulosT19} and \cite{PapadopoulosT20},
	asking whether there is an \XP-time algorithm for SFVS parameterized by the mim-width of a given decomposition. For rank-width and $\bQ$-rank-width
	we provide the first explicit \FPT-algorithms with low exponential dependency that avoid the MSO$_1$ formulation.
	Our main results are summarized in the following theorem:

	\begin{theorem}\label{theo:intromain}
		Let $G$ be a graph on $n$ vertices.
		We can solve \textsc{Subset Feedback Vertex Set} and \textsc{Node Multiway Cut} in time $2^{O(rw^3)}\cdot n^4$ and $2^{O(q^2\log(q))}\cdot n^4$, where $rw$ and $q$ are the rank-width and the $\bQ$-rank-width of $G$, respectively.
		Moreover, if a branch-decomposition of mim-width $k$ for $G$ is given as input, we can solve \textsc{Subset Feedback Vertex Set} and \textsc{Node Multiway Cut} in time $n^{O(k^2)}$.
	\end{theorem}
	
	Note it is not known whether the mim-width of a graph can be approximated within a constant factor in time  $n^{f(k)}$ for some function $f$.
	However, by the previously mentioned results of Belmonte and Vatshelle \cite{BelmonteV13} on computing decompositions of bounded mim-width, combined with a result of \cite{JaffkeKST19tcs} showing that for any positive integer $r$ a decomposition of mim-width $k$ of a graph $G$ is also a decomposition of mim-width at most $2k$ of its power $G^r$, we get the following corollary.
	
	\begin{corollary}\label{cor:mainintro}
		We can solve \textsc{Subset Feedback Vertex Set} and \textsc{Node Multiway Cut} in polynomial time on \textsc{Interval}, \textsc{Permutation}, and \textsc{Bi-Interval} graphs,
		\textsc{Circular Arc} and \textsc{Circular
			Permutation} graphs, \textsc{Convex} graphs, \textsc{$k$-Polygon},
		\textsc{Dilworth-$k$} and \textsc{Co-$k$-Degenerate} graphs for fixed
		$k$, and on arbitrary powers of graphs in any of these classes.
		
	\end{corollary}
	
	Previously, such polynomial-time tractability was known only for SFVS and only on \textsc{Interval} and \textsc{Permutation} graphs \cite{PapadopoulosT19}.
	It is worth noticing that Theorem \ref{theo:intromain} implies also that we can solve \textsc{Subset Feedback Vertex Set} and \textsc{Node Multiway Cut} in polynomial time on \textsc{Leaf Power} if an intersection model is given as input (from which we can compute a decomposition of mim-width 1) \cite{BelmonteV13,JaffkeKST19tcs} and on \textsc{$H$-Graphs} for a fixed $H$ if  an $H$-representation is given as input (from which we can compute a decomposition of mim-width $2|E(H)|$) \cite{FominGR18}.
	
	\medskip
	\paragraph{\bf Our approach.}
	We give some intuition to our meta-algorithm, that will focus on \textsc{Subset Feedback Vertex Set}.
	Since NMC can be solved by adding a vertex $v$ of large weight adjacent to all terminals and solving SFVS with $S=\{v\}$, all within the same runtime as extending the given branch-decomposition to this new graph increases the width at most by one for all considered width measures.
	
	Towards achieving our goal, we use the $d$-neighbor equivalence, with $d=1$ and $d=2$, a notion introduced by Bui-Xuan et~al.~\cite{BuiXuanTV13}.
	Two subsets $X$ and $Y$ of $A \subseteq V(G)$ are \emph{$d$-neighbor equivalent \wrt $A$}, if $\min(d,|X\cap N(u)|) = \min(d,|Y\cap N(u)|)$ for all $u\in V(G) \setminus A$.
	For a cut $(A, \comp{A})$ this equivalence relation on subsets of vertices was used by Bui-Xuan et~al.~\cite{BuiXuanTV13} to design a meta-algorithm, also giving \XP algorithms by mim-width, for so-called $(\sigma, \rho)$ generalized domination problems.
	Recently, Bergougnoux and Kant\'e \cite{BergougnouxK19esa} extended the uses of this notion to acyclic and connected variants of $(\sigma, \rho)$ generalized domination and similar problems like FVS. An earlier \XP algorithm for FVS parameterized by mim-width had been given by Jaffke et~al.~\cite{JaffkeKT20} but instead of the $d$-neighbor equivalences this algorithm was based on the notions of reduced forests and minimal vertex covers.
	
	Our meta-algorithm does a bottom-up traversal of a given branch-decomposition of the input graph $G$, computing a vertex subset $X$ of maximum weight that induces an \emph{$S$-forest} (\ie, a graph where no cycle contains a vertex of $S$) and outputs $V(G) \setminus X$ which is necessarily a solution of SFVS.
	As usual, our dynamic programming algorithm  relies on a notion of representativity between sets of partial solutions.
	For each cut $(A,\comp{A})$ induced by the decomposition, our algorithm computes a set of partial solutions $\cA \subseteq 2^{A}$ of small size that represents $2^{A}$. We say that a set of partial solutions $\cA\subseteq 2^{A}$ represents a set of partial solutions $\cB\subseteq 2^{A}$, if, for each $Y\subseteq \comp{A}$, we have $\best(\cA,Y)=\best(\cB,Y)$ where $\best(\cC,Y)$ is the maximum weight of a set $X\in\cC$ such that $X\cup Y$ induces an $S$-forest.
	Since the root of the decomposition is associated with the cut $(V(G),\emptyset)$, the set of partial solutions computed for this cut represents $2^{V(G)}$ and thus contains an $S$-forest of maximum weight.
	Our main tool is a subroutine that, given a set of partial solutions $\cB\subseteq 2^{A}$, outputs a subset $\cA\subseteq \cB$ of small size that represents $\cB$.

	To design this subroutine, we cannot use directly the approaches solving FVS of any earlier approaches, like \cite{BergougnouxK19esa} or \cite{JaffkeKT20}.
	This is due to the fact that $S$-forests behave quite differently than forests; for example, given an $S$-forest $F$, the graph induced by the edges between $A\cap V(F)$ and $\comp{A}\cap V(F)$ could be a biclique.
	Instead, we introduce a notion of vertex contractions and prove that, for every $X\subseteq A$ and $Y\subseteq \comp{A}$, the graph induced by $X\cup Y$ is an $S$-forest if and only if there exists a partition of $X \setminus S$ and of $Y \setminus S$, satisfying certain properties, such that contracting the blocks of these partitions into single vertices transforms the $S$-forest into a forest.
	
	This equivalence between $S$-forests in the original graph and forests in the contracted graphs allows us to adapt some ideas from \cite{BergougnouxK19esa} and \cite{JaffkeKT20}.
	Most of all, we use the property that, if the mim-width of the given decomposition is $mim$, then the contracted graph obtained from the bipartite graph induced by $X$ and $Y$ admits a vertex cover $\VC$ of size at most $4mim$.
	Note however, that in our case the elements of $\VC$ are contracted subsets of vertices. Such a vertex cover allows us to control the cycles which are crossing the cut.
	
	We associate each possible vertex cover $\VC$ with an index $i$ which contains basically a representative for the 2-neighbor equivalence for each subset of vertices in $\VC$.
	Moreover, for each index $i$, we introduce the notions of partial solutions and complements solutions associated with $i$ which correspond, respectively, to subsets of $X\subseteq A$ and subsets $Y\subseteq \comp{A}$ such that, for some contractions of $X$ and $Y,$ the contracted graph obtained from the bipartite graph induced by $X$ and $Y$  admits a vertex cover $\VC$ associated with $i$.
	We define an equivalence relation $\sim_i$ between the partial solutions associated with $i$ such that $X\sim_i W$, if $X$ and $W$ connect in the same way the representatives of the vertex sets which belongs to the vertex covers described by $i$.
	Given a set of partial solutions $\cB\subseteq 2^{A}$, our subroutine outputs a set $\cA$ that contains, for each index $i$ and each equivalence class $\cC$ of $\sim_i$ over $\cB$, a partial solution in $\cC$ of maximum weight.
	In order to prove that $\cA$ represents $\cB$, we show that:
	\begin{itemize}
		\item for every $S$-forest $F$, there exists an index $i$ such that $V(F)\cap A$ is a partial solution associated with $i$ and $V(F)\cap \comp{A}$ is a complement solutions associated with $i$.
		
		\item if $X\sim_i W$, then, for every complement solution $Y$ associated with $i$, the graph induced by $X\cup Y$ is an $S$-forest if and only if $W\cup Y$ induces an $S$-forest.
	\end{itemize}
	The number of indices $i$ is upper bounded by $2^{O(q^2\log(q))}$, $2^{O(rw^3)}$ and $n^{O(mim^2)}$.
	This follows from the known upper-bounds on the number of 2-neighbor equivalence classes and the fact that the vertex covers we consider have size at most $4mim$.
	Since there are at most $(4mim)^{4mim}$ ways of connecting $4mim$ vertices and $rw,q\geq mim$, we deduce that the size of $\cA$ is upper bounded by $2^{O(q^2\log(q))}$, $2^{O(rw^3)}$ and $n^{O(mim^2)}$.
	
	\medskip
	
	To the best of our knowledge, this is the first time a dynamic programming algorithm parameterized by graph width measures uses this notion of vertex contractions.
	Note that in contrast  to the meta-algorithms in \cite{BergougnouxK19esa,BuiXuanTV13}, the number of representatives (for the $d$-neighbor equivalence) contained in the indices of our meta-algorithm are not upper bounded by a constant but by $4mim$.
	This explains the differences between the runtimes in Theorem~\ref{theo:intromain} and those obtained in~\cite{BergougnouxK19esa,BuiXuanTV13}, i.e. $n^{O(mim^2)}$ versus $n^{O(mim)}$.
	However, for the case $S=V(G)$, thus solving FVS, our meta-algorithm will have runtime $n^{O(mim)}$, as the algorithms for FVS of~\cite{BergougnouxK19esa, JaffkeKT20}.
	We do not expect that SFVS can be solved as fast as FVS when parameterized by graph width measures.
	In fact, we know that it is not the case for tree-width as FVS can be solved in $2^{O(k)} \cdot n$~\cite{BodlaenderCKN15} but SFVS cannot be solved in $k^{o(k)} \cdot n^{O(1)}$ unless ETH fails~\cite{BergougnouxBBK20}.

	\section{Preliminaries }
	The size of a set $V$ is denoted by $|V|$ and its power set is denoted by $2^V$.
	We write $A\setminus B$ for the set difference of $A$ from $B$.
	We let $\min (\emptyset)= +\infty$ and $\max(\emptyset)=-\infty$.
	
	\paragraph{\bf  Graphs}
	The vertex set of a graph $G$ is denoted by $V(G)$ and its edge set by $E(G)$.
	An edge between two vertices $x$ and $y$ is denoted by $xy$ (or $yx$).
	Given $\cS\subseteq 2^{V(G)}$, we denote by $V(\cS)$ the set $\bigcup_{S\in\cS}S$.
	For a vertex set $U \subseteq V(G)$, we denote by $\comp{U}$ the set $V(G) \setminus U$.
	The set of vertices that are adjacent to $x$ is denoted by $N_G(x)$, and for $U\subseteq V(G)$, we let $N_G(U)=\left(\cup_{v\in U}N_G(v)\right)\setminus U$.
	
	
	The subgraph of $G$ induced by a subset $X$ of its vertex set is denoted by $G[X]$.
	For two disjoint subsets $X$ and $Y$ of $V(G)$, we denote by $G[X,Y]$ the bipartite graph with vertex set $X\cup Y$ and edge set $\{xy \in E(G)\mid x\in X \text{ and } \ y\in Y \}$. We denote by $M_{X,Y}$ the adjacency matrix between $X$ and $Y$, \ie, the $(X,Y)$-matrix such that $M_{X,Y}[x,y]=1$ if $y\in N(x)$ and 0 otherwise.
	A \emph{vertex cover} of a graph $G$ is a set of vertices $\VC\subseteq V(G)$ such that, for every edge $uv\in E(G)$, we have $u\in\VC$ or $v\in \VC$.
	A \emph{matching} is a set of edges having no common endpoint
	and an \emph{induced matching} is a matching $M$ of edges such that $G[V(M)]$ has no other edges besides $M$.
	The \emph{size of an induced matching} $M$ refers to the number of edges in $M$.
	
	For a graph $G$, we denote by $\cc_G(X)$ the partition $\{C\subseteq V(G) \mid G[C]$ is a connected component of $G[X]\}$.
	We will omit the subscript $G$ of the neighborhood and components notations whenever there is no ambiguity.
	
	For two graphs $G_1$ and $G_2$,	we denote by $G_1 - G_2$ the graph $(V(G_1), E(G_1)\setminus E(G_2))$.
	
	Given a graph $G$ and $S\subseteq V(G)$, we say that  a cycle of $G$ is an \emph{$S$-cycle} if it contains a vertex in $S$.
	Moreover, we say that a subgraph $F$ of $G$ is an \emph{$S$-forest} if $F$ does not contain an $S$-cycle.
	Typically, the \textsc{Subset Feedback Vertex Set} problem asks for a vertex set of minimum (weight) size such that its removal results in an $S$-forest.
	Here we focus on the following equivalent formulation:
	\pbDef{Subset Feedback Vertex Set (SFVS)}
	{A graph $G$, $S\subseteq V(G)$ and a weight function $\w : V(G)\to \bQ$}
	{The maximum among the weights of the $S$-forests of $G$.}
	
	\paragraph{\bf Rooted Layout}
	For the notion of branch-decomposition, we consider its rooted variant called \emph{rooted layout}.
	A \emph{rooted binary tree} is a binary tree with a distinguished vertex called the \emph{root}.
	Since we manipulate at the same time graphs and trees representing them,  the vertices of trees will be called \emph{nodes}.
	
	A rooted layout of $G$ is a pair $(T,\delta)$ of a rooted binary tree $T$ and a bijective function $\delta$ between $V(G)$ and the leaves of $T$.
	For each node $x$ of $T$, let $L_x$ be the set of all the leaves $l$ of $T$ such that the path from the root of $T$ to $l$ contains $x$.
	We denote by $V_x$ the set of vertices that are in bijection with $L_x$, \ie, $V_x:=\{v \in V(G)\mid \delta(v)\in L_x\}$.
	
	All the width measures dealt with in this paper are special cases of the following one, where the difference in each case is the used set function.
	Given a set function $\f: 2^{V(G)} \to\bN$ and a rooted layout $\cL=(T,\delta)$, the $\f$-width of a node $x$ of $T$ is $\f(V_x)$ and the $\f$-width of $(T,\delta)$, denoted by $\f(T,\delta)$ (or $\f(\cL)$), is $\max\{\f(V_x) \mid x \in V(T)\}$.
	Finally, the $\f$-width of $G$ is the minimum $\f$-width over all rooted layouts of $G$.
	
	\paragraph{\bf $(\bQ)$-Rank-width}
	The rank-width and $\bQ$-rank-width are, respectively, the $\rw$-width and $\Qrw$-width where $\rw(A)$ (resp. $\Qrw(A)$) is the rank over $GF(2)$ (resp. $\bQ$) of the matrix $M_{A,\comp{A}}$ for all $A\subseteq V(G)$.
	
	\paragraph{\bf Mim-width}
	The mim-width of a graph $G$ is the $\mim$-width of $G$ where $\mim(A)$ is the size of a maximum induced matching of the graph $G[A,\comp{A}]$ for all $A\subseteq V(G)$.
	\medskip
	
	
	Observe that all three parameters $\rw$-, $\Qrw$-, and $\mim$-width are symmetric, \ie, for the associated set function $f$ and for any $A \subseteq V(G)$, we have $f(A) = f(\comp{A})$.
	The following lemma provides upper bounds between mim-width and the other two parameters. 
	
	\begin{lemma}[\cite{Vatshelle12}]\label{lem:comparemim}
		Let $G$ be a graph. For every $A\subseteq V(G)$,  
		we have $\mim(A) \leq \rw(A)$ and $\mim(A) \leq \Qrw(A)$.
	\end{lemma}
	\begin{proof}
		Let $A\subseteq V(G)$.
		Let $S$ be the vertex set of a maximum induced matching of the graph $G[A,\comp{A}]$.
		By definition, we have $\mim(A)= |S\cap A| = |S \cap \comp{A}|$.
		Observe that the restriction of the matrix $M_{A,\comp{A}}$ to rows in $S \cap A$ and columns in $S \cap \comp{A}$ is a permutation matrix: a binary square matrix with exactly one entry of 1 in each row and each column. The rank of this permutation matrix over $GF[2]$ or $\bQ$ is $|S\cap A|=\mim(A)$.
		Hence, $\mim(A)$ is upper bounded both by $\rw(A)$ and $\Qrw(A)$.
	\end{proof}

	\paragraph{\bf $d$-neighbor-equivalence.} The following concepts were introduced in \cite{BuiXuanTV13}. Let $G$ be a graph.
	Let $A\subseteq V(G)$ and $d\in \bN^+$.
	Two subsets $X$ and $Y$ of $A$ are \emph{$d$-neighbor equivalent \wrt $A$}, denoted by $X\equi{A}{d} Y$, if $\min(d,|X\cap N(u)|) = \min(d,|Y\cap N(u)|)$ for all $u\in \comp{A}$.
	It is not hard to check that $\equi{A}{d}$ is an equivalence relation.
	See Figure \ref{fig:nec} for an example of $2$-neighbor equivalent sets.
	
	\begin{figure}[h!]
		\centering
		\includegraphics[scale=0.95]{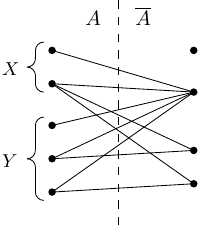}
		\caption{We have $X\equi{A}{2} Y$, but it is not the case that $X\equi{A}{3} Y$.}
		\label{fig:nec}
	\end{figure}

	For all $d\in \bN^+$, we let $\nec_d : 2^{V(G)}\to \bN$ where for all $A\subseteq V(G)$, $\nec_d(A)$ is the number of equivalence classes of $\equi{A}{d}$.
	Notice that $\nec_1$ is a symmetric function \cite[Theorem 1.2.3]{Kim82} but $\nec_d$ is not necessarily symmetric for $d\geq 2$.
	To simplify the running times, we will use the shorthand $\snec_2(A)$ to denote $\max(\nec_2(A),\nec_2(\comp{A}))$ (where $\mathsf{s}$ stands for symmetric).
	The following lemma shows how $\nec_d(A)$ is upper bounded by the other parameters. 
	
	\begin{lemma}[\cite{BelmonteV13,OumSV13,Vatshelle12}]\label{lem:compare}
		Let $G$ be a graph. For every $A\subseteq V(G)$ and $d\in \bN^+$, we have the following upper bounds on $\nec_d(A)$:
		\begin{multicols}{3}
			\begin{enumerate}[(a)]
				\item $2^{d \rw(A)^2}$,
				\item $2^{\Qrw(A)\log(d \Qrw(A) + 1 )}$,
				\item $|A|^{d \mim(A)}$.
			\end{enumerate}
		\end{multicols}
	\end{lemma}
	
	In order to manipulate the equivalence classes of $\equi{A}{d}$, one needs to compute a representative  for each equivalence class in polynomial time.
	This is achieved with the following notion of a representative.
	Let $G$ be a graph with an arbitrary ordering of $V(G)$ and let $A\subseteq V(G)$.
	For each $X\subseteq A$, let us denote by $\rep{A}{d}(X)$ the lexicographically smallest set $R\subseteq A$ such that $|R|$ is minimized and $R\equi{A}{d} X$.
	Moreover, we denote by $\Rep{A}{d}$ the set $\{\rep{A}{d}(X)\mid X\subseteq A\}$.
	It is worth noticing that the empty set always belongs to $\Rep{A}{d}$, for all $A\subseteq V(G)$ and $d\in\bN^+$.
	Moreover, we have $\Rep{V(G)}{d}=\Rep{\emptyset}{d}=\{\emptyset\}$ for all $d\in \bN^+$.
	In order to compute these representatives, we use the following lemma.
	
	\begin{lemma}[\cite{BuiXuanTV13}]\label{lem:computenecd}
		Let $G$ be an $n$-vertex graph. For every $A\subseteq V(G)$ and $d\in \bN^+$, one can compute in time $O(\nec_d(A) \cdot n^2 \cdot \log(\nec_d(A)))$, the sets $\Rep{A}{d}$ and a data structure that, given
		a set $X\subseteq A$, computes $\rep{A}{d}(X)$ in time $O(|A|\cdot n\cdot \log(\nec_d(A)))$.
	\end{lemma}

	\paragraph{\bf Vertex Contractions}
	In order to deal with SFVS, we will use the ideas of the algorithms for \textsc{Feedback Vertex Set} from \cite{BergougnouxK19esa,JaffkeKT18}.
	To this end, we will contract subsets of $\comp{S}$ in order to transform $S$-forests into forests.
	
	In order to compare two partial solutions associated with $A\subseteq V(G)$, we define an auxiliary graph in which we replace contracted vertices by their representative sets in $\Rep{A}{2}$.
	Since the sets in $\Rep{A}{2}$ are not necessarily pairwise disjoint, we will use the following notions of graphs ``induced'' by collections of subsets of vertices. We will also use these notions to define the contractions we make on partial solutions.
	
	\medskip
	
	Let $G$ be a graph.
	Given $\cA\subseteq 2^{V(G)}$, we define $G[\cA]$ as the graph with vertex set $\cA$ where $A,B\in \cA$ are adjacent if and only if $N(A)\cap B\neq\emptyset$.
	Observe that if the sets in $\cA$ are pairwise disjoint, then $G[\cA]$ is obtained from an induced subgraph of $G$ by \emph{vertex contractions} (i.e., by replacing two vertices $u$ and $v$ with a new vertex with neighborhood $N(\{u,v\})$) and, for this reason, we refer to $G[\cA]$ as a \emph{contracted graph}.
	Notice that we will never use the neighborhood notation and connected component notations on contracted graphs.
	Given $\cA,\cB\subseteq 2^{V(G)}$, we denote by $G[\cA,\cB]$ the bipartite graph with vertex set $\cA\cup\cB$ and where $A,B\in \cA\cup \cB$ are adjacent if and only if $A\in \cA$, $B\in \cB$, and $N(A)\cap B\neq\emptyset$.
	Moreover, we denote by $G[\cA \mid \cB]$ the graph with vertex set $\cA\cup\cB$ and with edge set $E(G[\cA])\cup E(G[\cA,\cB])$.
	Observe that both graphs $G[\cA,\cB]$ and $G[\cA \mid \cB]$ are subgraphs of the contracted graph $G[\cA \cup \cB]$.
	To avoid confusion with the original graph, we refer to the vertices of the contracted graphs as \textit{blocks}.
	It is worth noticing that in the contracted graphs used in this paper, whenever two blocks are adjacent, they are disjoint.
	
	\medskip
	
	The following observation states that we can contract from a partition without increasing the size of a maximum induced matching of a graph.
	It follows directly from the definition of contractions.
	
	\begin{observation}\label{obs:contactionmimwidth}
		Let $H$ be a graph.
		For any partition $\cP$ of a subset of $V(H)$, the size of a maximum induced matching of $H[\cP]$ is at most the size of a maximum induced matching of $H$.
	\end{observation}
	
	Let $(G,S)$ be an instance of SFVS.
	The vertex contractions that we use on a partial solution $X$ are defined from a given partition of $X\setminus S$.
	A partition of the vertices of $X \setminus S$ is called an \emph{$\comp{S}$-contraction} of $X$.
	We will use the following notations to handle these contractions.
	
	Given $Y\subseteq V(G)$, we denote by $\binom{Y}{1}$ the partition of $Y$ which contains only singletons, i.e., $\binom{Y}{1}=\{\{ v\} \mid v \in Y\}$.
	Moreover, for an $\comp{S}$-contraction $\cP$ of $X$, we denote by $X_{\downarrow \cP}$ the partition of $X$ where $X_{\downarrow\cP}= \cP \cup \binom{X\cap S}{1}$.
	Given a subgraph $H$ of $G$ and an $\comp{S}$-contraction $\cP$ of $V(H)$, we denote by $H_{\downarrow \cP}$ the graph $H[V(H)_{\downarrow \cP}]$.
	For example, given two $\comp{S}$-contractions $\cP_X,\cP_Y$ of two disjoint subsets $X,Y$ of $V(G)$, we denote the graph $G[X_{\downarrow \cP_X}, Y_{\downarrow \cP_Y}]$ by $G[X, Y]_{\downarrow\cP_X\cup \cP_Y}$ and the graph $G[X_{\downarrow \cP_X} | Y_{\downarrow \cP_Y}]$ by $G[X| Y]_{\downarrow\cP_X\cup \cP_Y}$.
	It is worth noticing that in our contracted graphs, all the blocks of $S$-vertices are singletons and we denote them by $\{v\}$.
	
	Given a set $X\subseteq V(G)$, we will intensively use the graph $G[X]_{\downarrow\cc(X\setminus S)}$ which corresponds to the graph obtained from $G[X]$ by contracting the connected components of $G[X\setminus S]$, see Figure~\ref{fig:S-contraction}.
	
	\begin{figure}[htb]
		\centering
		\includegraphics[width=0.6\linewidth]{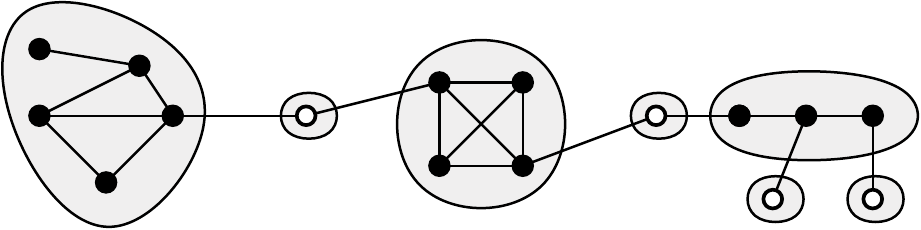}
		\caption{An $S$-forest induced by a set $X\subseteq V(G)$, the vertices of $S$ are white. The gray circles represent the blocks of $X_{\downarrow \cc(X\setminus S)}$. The graph $G[X]_{\downarrow\cc(X\setminus S)}$ is obtained by contracting each gray circle.}
		\label{fig:S-contraction}
	\end{figure}

	Observe that, for every subset $X\subseteq V(G)$, if $G[X]$ is an $S$-forest, then $G[X]_{\downarrow\cc(X\setminus S)}$ is a forest.
	The converse is not true as we may delete $S$-cycles with contractions: take a triangle with one vertex $v$ in $S$ and contract the neighbors of $v$.
	However, we can prove the following equivalence.

	\begin{fact}\label{fact:Sforestandforest}
		Let $G$ be a graph and $S\subseteq V(G)$.
		For every $X\subseteq V(G)$, $G[X]$ is an $S$-forest if and only if there exists an $\comp{S}$-contraction $\cP$ of $X$ satisfying the following two properties:
		\begin{itemize}
			\item $G[X]_{\downarrow \cP}$ is a forest, and
			\item for every $B\in \cP$ and $v\in X\cap S$, we have $|N(v)\cap B|\leq 1$.
		\end{itemize}
		Moreover, if $G[X]$ is an $S$-forest, then the $\comp{S}$-contraction $\cc(X\setminus S)$ satisfies these two properties. 
	\end{fact}	
	\begin{proof}
		($\Rightarrow$) Suppose first that $G[X]$ is an $S$-forest. We claim that the $\comp{S}$-contraction $\cc(X\setminus S)$ satisfies the two properties.
		Assume towards a contradiction that there is a cycle $C$ in $G[X]_{\downarrow\cc(X\setminus S)}$.
		By definition of $G[X]_{\downarrow\cc(X\setminus S)}$, the blocks of this graph are the connected components $\cc(X\setminus S)$ and the singletons in $\binom{X\cap S}{1}$.
		The blocks in $\cc(X\setminus S)$ are pairwise non-adjacent, thus $C$ contains a block $\{s\}$ with $s\in X\cap S$.
		Observe that for every pair of consecutive blocks $B_1,B_2$ of $C$, there exists a vertex $v_1\in B_1$ and $v_2\in B_2$ such that $v_1v_2\in E(G[X])$.
		As every block of $G[X]_{\downarrow \cc(X\setminus S)}$ induces a connected component in $G$, we can construct an $S$-cycle in $G[X]$ by replacing every block of $C$ by a path in $G[X]$, yielding a contradiction.
		Hence, $G[X]_{\downarrow\cc(X\setminus S)}$ is a forest, i.e. the first property is satisfied.
		Observe that if there exists $C\in \cc(X\setminus S)$ and $v\in X\cap S$ such that $v$ has two neighbors in $C$, then there exists an $S$-cycle in $G[X]$ since $C$ is a connected component.
		Hence, $\cc(X\setminus S) $ satisfies the second property.
		
		($\Leftarrow$) Let $\cP$ be a $\comp{S}$ contraction of a subset $X\subseteq V(G)$ that satisfies the two properties.
		Assume for contradiction that there is an $S$-cycle $C$ in $G[X]$.
		Let $v$ be a vertex of $C$ that belongs to $S$ and let $u$ and $w$ be the neighbors of $v$ in $C$.
		
		Let $B_u$ and $B_w$ be the blocks in $(X\cup Y)_{\downarrow \cP}$ that contain $u$ and $w$ respectively.
		As $v$ is in $S$,  it belongs to the block $\{v\}$ of $G[X\cup Y]_{\downarrow \cP}$ and thus it is not contained in $B_u$ nor $B_w$.
		The second property implies that $|N(v)\cap B|\leq 1$ for each $B\in \cP$.
		Thus, $B_u$ and $B_w$ are two distinct blocks both connected to the block $\{v\}$.
		Since there exists a path between $u$ and $w$ in $C$ that does not go through $v$, we deduce that there is a path between $B_u$ and $B_w$ in $G[X]_{\downarrow \cP}$ that does not go through $\{v\}$.
		Indeed, this follows from the fact that if there is an edge between two vertices $a$ and $b$ in $G[X]$, then either $a$ and $b$ belong to the same block of $G[X]_{\downarrow\cP}$ or there exists an edge between the blocks in $G[X]_{\downarrow\cP}$ which contain $a$ and $b$.
		We conclude that there exists a cycle in $G[X]_{\downarrow \cP}$, a contradiction with the first property.
	\end{proof}

	\section{A Meta-Algorithm for Subset Feedback Vertex Set}
	
	In the following, we present a meta-algorithm that, given a rooted layout $(T,\delta)$ of $G$, solves \textsc{SFVS}.
	We will show that this meta-algorithm will imply that \textsc{SFVS} can be solved in time $2^{O(\Qrw(G)^2\log(\Qrw(G)))}\cdot n^{4}$, $2^{O(\rw(G)^3)}\cdot n^{4}$ and $n^{O(\mim(T,\delta)^2)}$.
	The main idea of this algorithm is to use $\comp{S}$-contractions in order to employ similar properties of the algorithm for \textsc{Maximum Induced Tree} of \cite{BergougnouxK19esa} and the $n^{O(\mim(T,\delta))}$ time algorithm for \textsc{Feedback Vertex Set} of \cite{JaffkeKT18}.
	In particular, we use the following lemma which is proved implicitly in \cite{BergougnouxK19esa}.
	To simplify the following statements, we fix a graph $G$, a rooted layout $(T,\delta)$ of $G$ and a node $x\in V(T)$.
	
	\begin{lemma}\label{lem:X2+}
		Let $X$ and $Y$ be two disjoint subsets of $V(G)$.
		If $G[X\cup Y]$ is a forest, then the number of vertices of $X$ that have at least two neighbors in $Y$ is bounded by $2w$ where $w$ is the size of a maximum induced matching in the bipartite graph $G[X,Y]$.
	\end{lemma}
	\begin{proof}
		Let $X^{2+}$ be the set of vertices in $X$ having at least $2$ neighbors in $Y$.
		In the following, we prove that $F=G[X^{2+},Y]$ admits a \emph{good bipartition}, that is a bipartition $\{X_0,X_1\}$ of $X^{2+}$ such that, for each $i\in\{0,1\}$ and, for each $v\in X_i$, there exists $y_v\in Y$ such that $N_{F}(y_v)\cap X_i = \{v\}$.
		Observe that this is enough to prove the lemma since if $F$ admits a good bipartition $\{X_0,X_1\}$, then $|X_0|\leq w$ and $|X_1|\leq w$. 	
		Indeed, if $F$ admits a good bipartition $\{X_0,X_1\}$, then, for each $i\in\{0,1\}$, the set of edges $M_i=\{ vy_v \mid v\in X_i \}$ is an induced matching of $G[X,Y]$.
		In order to prove that $F$ admits a good bipartition it is sufficient to prove that each connected component of $F$ admits a good bipartition.
		
		Let $C$ be a connected component of $F$ and $x_0\in C\cap X^{2+}$. As $G[X\cup Y]$ is a forest, we deduce that $F[C]$ is a tree.
		Observe that the distance in $F[C]$ between each vertex $v\in C\cap X^{2+}$ and $x_0$ is even because $F[C]$ is bipartite \wrt $(C\cap X^{2+}, C\cap Y)$.
		Let $X_0$ (resp. $X_1$) be the set of all vertices $v\in C\cap X^{2+}$ at distance $2\ell$ from $x_0$ with $\ell$ even (resp. odd).
		We claim that $\{X_0,X_1\}$ is a good bipartition of $F[C]$.
		
		Let $i \in \{0,1\}$, $v\in X_i$ and $\ell\in \bN$ such that the distance between $v$ and $x_0$ in $F[C]$ is $2\ell$.
		Let $P$ be the set of vertices in $C\setminus \{v\}$ that share a common neighbor with $v$ in $F[C]$.
		We want to prove that $v$ has a neighbor $y$ in $F$ that is not adjacent to $P\cap X_i$.
		Observe that, for every $v'\in P$, the distance between $v'$ and $x_0$ in $F[C]$ is either $2\ell-2$, $2\ell$ or $2\ell+2$ because $F[C]$ is a tree and the distance between $v$ and $x_0$ is $2\ell$.
		By construction of $\{X_0,X_1\}$, every vertex at distance $2\ell-2$ or $2\ell+2$ from $x_0$ belongs to $X_{1-i}$.
		Thus, every vertex in $P\cap X_i$ is at distance $2\ell$ from $x_0$.
		If $\ell=0$, then we are done because $v=x_0$ and $P\cap X_i=\emptyset$.
		Assume that $\ell\neq 0$.
		As $F[C]$ is a tree, $v$ has only one neighbor $w$ at distance $2\ell-1$ from $x_0$ in $F[C]$.
		Because $F[C]$ is a tree, $w$ is the only vertex adjacent to $v$ and the vertices in $P\cap X_i$.
		By definition of $X^{2+}$, $v$ has at least two neighbors in $Y$, so $v$ admits a neighbor that is not $w$ and this neighbor is not adjacent to the vertices in $P\cap X_i$.
		Hence, we deduce that $\{X_0,X_1\}$ is a good bipartition of $F[C]$.
		
		We deduce that every connected component of $F$ admits a good bipartition and, thus, $F$ admits a good bipartition. This proves that $|X^{2+}|\leq 2w$.
	\end{proof}

	The following lemma generalizes Fact~\ref{fact:Sforestandforest} and presents the equivalence between $S$-forests and forests that we will use in our algorithm.
	
	\begin{lemma}\label{lem:equiStreeScontrac}
		Let $X\subseteq V_x$ and $Y\subseteq \comp{V_x}$.
		If the graph $G[X\cup Y]$ is an $S$-forest, then there exists an $\comp{S}$-contraction $\cP_Y$ of $Y$ that satisfies the following conditions:
		\begin{enumerate}[(1)]
			\item $G[X\cup Y]_{\downarrow\cc(X\setminus S)\cup \cP_Y}$ is a forest,
			\item for all block $P\in\cc(X\setminus S)\cup \cP_Y$ and $v\in (X\cup Y)\cap S$, we have $|N(v)\cap P|\leq 1$,
			\item the graph $G[X, Y]_{\downarrow \cc(X\setminus S)\cup\cP_Y}$ admits a vertex cover $\VC$ of size at most $4 \mim(V_x)$ such that the neighborhoods of the blocks in $\VC$ are pairwise distinct in $G[X, Y]_{\downarrow \cc(X\setminus S)\cup\cP_Y}$.
		\end{enumerate}
	\end{lemma}
	\begin{proof}	
		Assume that $G[X\cup Y]$ is an $S$-forest.
		Let us explain how we construct $\cP_Y$ that satisfies Conditions (1)-(3).
		First, we initialize $\cP_Y=\cc(Y\setminus S)$.
		Observe  that there is no cycle in $G[X\cup Y]_{\downarrow \cc(X\setminus S)\cup \cP_Y}$ that contains a block in $\binom{S}{1}$ because $G[X\cup Y]$ is an $S$-forest.
		Moreover, $\cc(X\setminus S)$ and $\cP_Y$ form two independent sets in $G[X\cup Y]_{\downarrow \cc(X\setminus S)\cup \cP_Y}$.
		Consequently, for all the cycles $C$ in $G[X\cup Y]_{\downarrow \cc(X\setminus S)\cup \cP_Y}$ we have $C=(X_1,Y_1,X_2,Y_2,\dots,X_t,Y_t)$ where $X_1,\dots,X_t\in\cc(X\setminus S)$ and $Y_1,\dots,Y_t\in \cP_Y$.
		We do the following operation, until the graph $G[X\cup Y]_{\downarrow \cc(X\setminus S)\cup \cP_Y}$ is a forest: take a cycle $C=(X_1,Y_1,X_2,Y_2,\dots,X_t,Y_t)$ in $G[X\cup Y]_{\downarrow \cc(X\setminus S)\cup \cP_Y}$ and replace the blocks $Y_1,\dots,Y_t$ in $\cP_Y$ by the block $Y_1\cup\dots\cup Y_t$.
		See Figures~\ref{fig:killcycle} and~\ref{fig:contracted} in Appendix~\ref{appendix} for an example of $\comp{S}$-contraction $\cP_Y$.

		For each $B\in \cc(X\setminus S)\cup \cc(Y\setminus S)$, the vertices of $B$ are pairwise connected in $G[(X\cup Y)\setminus S]$. We deduce by induction that whenever we apply the operation on a cycle $C=(X_1,Y_1,X_2,Y_2,\dots,X_t,Y_t)$, it holds that the vertices of the new block $Y_1\cup\dots\cup Y_t$ are pairwise connected in $G[(X\cup Y)\setminus S]$.
		Thus, for every block $B$ of $\cc(X\setminus S)\cup \cP_Y$, the vertices of $B$ are pairwise connected in $G[(X\cup Y)\setminus S]$.
		It follows that for every $v\in (X\cup Y)\cap S$ and $B\in \cc(X\setminus S)\cup \cP_Y$, since $G[X\cup Y]$ is an $S$-forest, we have $|N(v)\cap B|\leq 1$. Thus, Condition (2) is satisfied.
		
		It remains to prove Condition (3).
		Let $\VC$ be the set of blocks of $G[X, Y]_{\downarrow \cc(X\setminus S)\cup \cP_Y}$ containing:
		\begin{itemize}
			\item the blocks that have at least 2 neighbors in $G[X, Y]_{\downarrow \cc(X\setminus S)\cup \cP_Y}$, and
			\item one block in every isolated edge of $G[X, Y]_{\downarrow \cc(X\setminus S)\cup \cP_Y}$.
		\end{itemize}
		By construction, it is clear that $\VC$ is indeed a vertex cover of $G[X, Y]_{\downarrow \cc(X\setminus S)\cup\cP_Y}$ as every edge is either isolated or incident to a block of degree at least 2.
		We claim that $|\VC|\leq 4\mim(V_x)$.
		By Observation \ref{obs:contactionmimwidth}, we know that the size of a maximum induced matching in $G[X, Y]_{\downarrow \cc(X\setminus S)\cup \cP_Y}$ is at most $\mim(V_x)$.
		Let $t$ be the number of isolated edges in $G[X, Y]_{\downarrow \cc(X\setminus S)\cup \cP_Y}$.
		Observe that the size of a maximum induced matching in the graph obtained from $G[X, Y]_{\downarrow \cc(X\setminus S)\cup \cP_Y}$ by removing isolated edges is at most $\mim(V_x)-t$.
		By Lemma \ref{lem:X2+}, we know that there are at most $4(\mim(V_x)-t)$ blocks that have at least 2 neighbors in $G[X, Y]_{\downarrow \cc(X\setminus S)\cup \cP_Y}$.
		We conclude that $|\VC|\leq 4\mim(V_x)$.
		
		Since $G[X, Y]_{\downarrow \cc(X\setminus S)\cup \cP_Y}$ is a forest, the neighborhoods of the blocks that have at least 2 neighbors must be pairwise distinct.
		We conclude from the construction of $\VC$ that the neighborhoods of the blocks of $\VC$ in $G[X,Y]_{\downarrow\cc(X\setminus S)\cup\cP_Y}$ are pairwise distinct.
		Hence, Condition (3) is satisfied.
	\end{proof}
	
	In the following, we will use Lemma \ref{lem:equiStreeScontrac} to design some sort of equivalence relation between partial solutions.
	To this purpose, we use the following set of tuples. We call each such tuple an index because it corresponds to an index into a table in a DP (dynamic programming) approach. We do this even though the presentation of our algorithm is not given by a standard DP description.

	\begin{definition}[$\bI_x$]\label{def:indices}
		We define the set $\bI_x$ of indices as the set of tuples
		\[ (\cX_{\vc}^{\comp{S}},\cX_{\vc}^{S}, \cX_{\comp{\vc}},\cY_{\vc}^{\comp{S}},\cY_{\vc}^{S}) \in
		2^{\Rep{{V_x}}{2}}\times 2^{\Rep{V_x}{1}}\times \Rep{V_x}{1}\times 2^{\Rep{\comp{V_x}}{2}} \times  2^{\Rep{\comp{V_x}}{1}}  \]
		such that $|\cX_{\vc}^{\comp{S}}|+|\cX_{\vc}^{S}|+|\cY_{\vc}^{\comp{S}}| + |\cY_{\vc}^{S}| \leq 4\mim(V_x)$.
	\end{definition}

	These sets of indices play a major role in our meta-algorithm, in particular, the sizes of these sets of indices appear in the runtime of our meta-algorithm. In fact, to prove the algorithmic consequences of our meta-algorithm for rank-width, $\bQ$-rank-width and mim-width, we show (Lemma~\ref{lem:consequences}) that the size of $\bI_x$ is upper bounded by $2^{O(\rw(V_x)^3)}$, $2^{O(\Qrw(V_x)^2\log(\Qrw(V_x)))}$ and $n^{O(\mim(V_x)^2)}$.
	
	In the following, we will define partial solutions associated with an index $i\in\bI_x$ (a partial solution may be associated with many indices).
	In order to prove the correctness of our algorithm (the algorithm itself will not use this concept), we will also define \textit{complement solutions} (the sets $Y\subseteq \comp{V_x}$ and their $\comp{S}$-contractions $\cP_Y$) associated with an index $i$.
	We will prove that, for every partial solution $X$ and complement solution $(Y,\cP_Y)$ associated with $i$, if the graph $G[X\cup Y]_{\downarrow \cc(X\setminus S)\cup\cP_Y}$ is a forest, then $G[X\cup Y]$ is an $S$-forest.
	
	\medskip
	
	Let us give some intuition on these indices by explaining how one index is associated with a solution, figures explaining this association and the purposes of the sets $\cX_{\vc}^{\comp{S}},\cX_{\vc}^{S}, \cX_{\comp{\vc}},\cY_{\vc}^{\comp{S}},\cY_{\vc}^{S}$ can be found in Appendix~\ref{appendix}.
	Let  $X\subseteq V_x$ and $Y\subseteq \comp{V_x}$ such that $G[X\cup Y]$ is an $S$-forest.
	Let $\cP_Y$ be the $\comp{S}$-contraction of $Y$ and $\VC$ be a vertex cover of  $G[X,Y]_{\downarrow\cc(X\setminus S)\cup\cP_Y}$ given by Lemma \ref{lem:equiStreeScontrac}.
	Then, $X$ and $Y$ are associated with $i=(\cX_{\vc}^{\comp{S}},\cX_{\vc}^{S}, \cX_{\comp{\vc}},\cY_{\vc}^{\comp{S}},\cY_{\vc}^{S})\in \bI_x$ such that:
	\begin{itemize}
		\item $\cX_{\vc}^{S}$ (resp. $\cY_{\vc}^{S}$) contains the representatives of the blocks $\{v\}$ in $\VC$ such that $v\in X\cap S$ (resp. $v\in Y\cap S$) \wrt the 1-neighbor equivalence over $V_x$ (resp. $\comp{V_x}$).
		We will only use the indices where $\cX_{\vc}^{S}$ contains representatives of singletons, in other words, $\cX_{\vc}^{S}$ is included in $\{ \rep{V_x}{1}(\{v\}) \mid v\in V_x\}$ which can be much smaller than $\Rep{V_x}{1}$.
		The same observation holds for $\cY_{\vc}^{S}$.
		In Definition \ref{def:indices}, we state that $\cX_{\vc}^{S}$ and  $\cY_{\vc}^{S}$ are, respectively, subsets of $2^{\Rep{V_x}{1}}$ and $2^{\Rep{\comp{V_x}}{1}}$, for the sake of simplicity.
		
		\item $\cX_{\vc}^{\comp{S}}$ (resp. $\cY_{\vc}^{\comp{S}}$) contains the representatives of the blocks in $\cc(X\setminus S)\cap \VC$ (resp. $\cP_Y\cap \VC$) w.r.t. the 2-neighbor equivalence relation over $V_x$ (resp. $\comp{V_x}$).
	
		\item $\cX_{\comp{\vc}}$ is the representative of $X\setminus V(\VC)$ (the set of vertices which do not belong to the vertex cover) w.r.t. the 1-neighbor equivalence over $V_x$.
	\end{itemize}
	Because the neighborhoods of the blocks in $\VC$ are pairwise distinct in $G[X,Y]_{\downarrow \cc(X\setminus S)\cup\cP_Y}$ (Property (3) of Lemma \ref{lem:equiStreeScontrac}), there is a one to one correspondence between the representatives in $\cX_{\vc}^{\comp{S}}\cup \cX_{\vc}^{S}\cup \cY_{\vc}^{\comp{S}}\cup \cY_{\vc}^{S}$ and the blocks in $\VC$.

	While $\cX_{\vc}^{\comp{S}}, \cX_{\vc}^{S},\cY_{\vc}^{\comp{S}},\cY_{\vc}^{S}$ describe $\VC$, the representative set $\cX_{\comp{\vc}}$ describes the neighborhood of the blocks of $X_{\downarrow \cc(X\setminus S)}$ which are not in $\VC$.
	The purpose of $\cX_{\comp{\vc}}$ is to make sure that, for every partial solution $X$ and complement solution $(Y,\cP_Y)$ associated with $i$, the set $\VC$ described by $\cX_{\vc}^{\comp{S}}, \cX_{\vc}^{S},\cY_{\vc}^{\comp{S}},\cY_{\vc}^{S}$ is a vertex cover of $G[X,Y]_{\downarrow \cc(X\setminus S)\cup \cP_Y}$.
	For doing so, it is sufficient to require that $Y\setminus V(\VC)$ has no neighbor in $\cX_{\comp{\vc}}$ for every complement solution $(Y,\cP_Y)$ associated with $i$.
	
	Observe that the sets $\cX_{\vc}^{\comp{S}}$ and $\cY_{\vc}^{\comp{S}}$ contain representatives for the $2$-neighbor equivalence.
	We need the 2-neighbor equivalence to control the $S$-cycles which might disappear after vertex contractions.
	To prevent this situation, we require, for example, that every vertex in $X\cap S$ has at most one neighbor in $\comp{R}$ for each $\comp{R}\in \cY_{\vc}^{\comp{S}}$.
	Thanks to the 2-neighbor equivalence, a vertex $v$ in $X\cap S$ has at most one neighbor in $\comp{R}\in \cY_{\vc}^{\comp{S}}$ if and only if $v$ has at most one neighbor in the block of $\cP_Y$ associated with $\comp{R}$.
	This property of the 2-neighbor equivalence is captured by the following fact.
	
	\begin{fact}\label{fact:2neighborequivalence}
		For every $A\subseteq V(G)$ and $B,P\subseteq A$, if $B\equi{A}{2} P$, then, for all $v\in \comp{A}$, we have $|N(v)\cap B|\leq 1$ if and only if $|N(v)\cap P|\leq 1$.
	\end{fact}
	
	\medskip
	
	In order to define partial solutions associated with $i$, we need the following notion of auxiliary graph.
	Given $X\subseteq V_x$ and  $i=(\cX_{\vc}^{\comp{S}},\cX_{\vc}^{S}, \cX_{\comp{\vc}},\cY_{\vc}^{\comp{S}},\cY_{\vc}^{S})\in\bI_x$, we write $\aux(X,i)$ to denote the graph
	\[ G[X_{\downarrow \cc(X\setminus S)}\mid \cY_{\vc}^{\comp{S}} \cup  \cY_{\vc}^{S}  ]. \]
	Observe that $\aux(X,i)$ is obtained from the graph induced by $X_{\downarrow \cc(X\setminus S)} \cup \cY_{\vc}^{\comp{S}} \cup \cY_{\vc}^{S}$ by removing the edges between the blocks from $ \cY_{\vc}^{\comp{S}} \cup \cY_{\vc}^{S}$. Figure \ref{fig:auxxi} illustrates an example of the graph $\aux(X,i)$ and the related notions. The figures in Appendix~\ref{appendix} explain the relations between an $S$-forest and these auxiliary graphs. 
	
	\begin{figure}[h]
		\centering
		\includegraphics[width=0.9\linewidth]{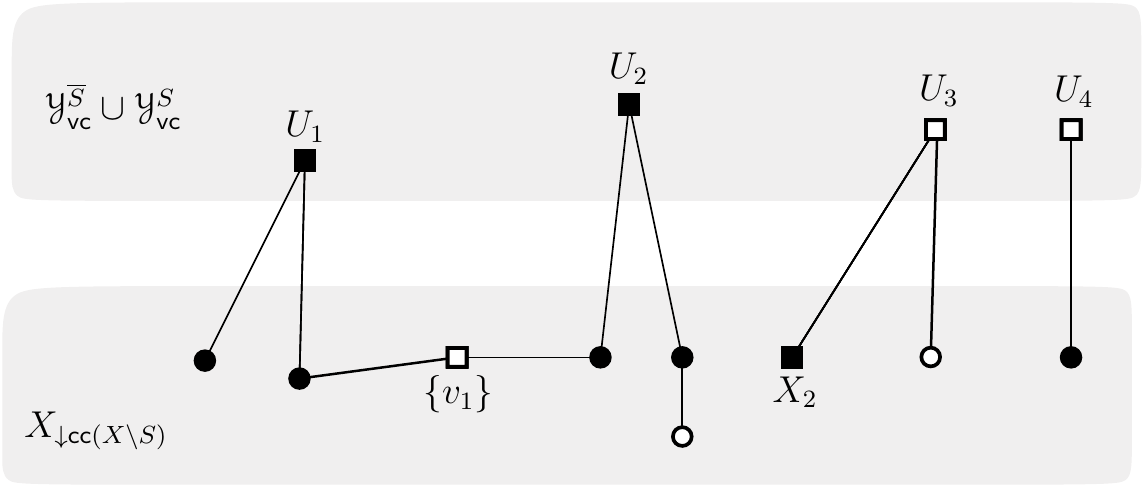}
		\caption{An example of a set $X\subseteq V_x$ and its auxiliary graph associated with the index $i=(\cX_{\vc}^{\comp{S}},\cX_{\vc}^{S}, \cX_{\comp{\vc}},\cY_{\vc}^{\comp{S}},\cY_{\vc}^{S})$.
		Here, $\cX_{\vc}^{\comp{S}}=\{R_2\}$ with $R_2=\rep{V_x}{2}(X_2)$, $\cX_{\vc}^{S}=\{R_1\}$ with $R_1= \rep{V_x}{1}(\{v_1\})$,  $\cX_{\comp{\vc}}$ is the representative of the union of the circular blocks, $\cY_{\vc}^{\comp{S}}=\{ U_1,U_2\}$, and $\cY_{\vc}^{S}=\{ U_3,U_4\}$. The singletons in $X\cap S$ and $\cY_{\vc}^{S}$ are white filled, whereas the square blocks are $\{v_1\}$, $X_2$ and the blocks  in $\cY_{\vc}^{\comp{S}}\cup \cY_{\vc}^{S}$.}
		\label{fig:auxxi}
	\end{figure}
	
	We will ensure that, given a complement solution $(Y,\cP_Y)$ associated with $i$, the graph $\aux(X,i)$ is isomorphic to $G[X_{\downarrow \cc(X\setminus S)} \mid Y_{\downarrow\cP_Y}\cap \VC]$.
	We are now ready to define the notion of partial solution associated with an index $i$. 
	
	\begin{definition}[Partial solutions]\label{def:partialsolution}
		Let $i=(\cX_{\vc}^{\comp{S}},\cX_{\vc}^{S}, \cX_{\comp{\vc}},\cY_{\vc}^{\comp{S}},\cY_{\vc}^{S})\in \bI_x$.
		We say that $X\subseteq V_x$ is a partial solution associated with $i$ if the following conditions are satisfied:
		\begin{enumerate}[(a)]
			
			\item  for every $R\in \cX_{\vc}^{S}$, there exists a unique $v\in X\cap S$ such that $R\equi{V_x}{1} \{v\}$,
			
			\item for every $R\in \cX_{\vc}^{\comp{S}}$, there exists a unique $C\in\cc(X\setminus S)$ such that $R\equi{V_x}{2} C$,
			
			\item  $\aux(X,i)$ is a forest,
			
			\item for every $C\in\cc(X\setminus S)$ and $\{v\}\in \cY_{\vc}^{S}$, we have $|N(v)\cap C |\leq 1$,
			
			\item for every $v\in X\cap S$ and $U\in \cY_{\vc}^{\comp{S}}\cup \cc(X\setminus S)$, we have $|N(v)\cap U|\leq 1$,
			
			\item $\cX_{\comp{\vc}} \equi{V_x}{1} X\setminus V(\VC_X)$ , where $\VC_X$ contains the blocks $\{v\}\in \binom{X\cap S}{1}$ such that $\rep{V_x}{1}(\{v\})\in \cX_{\vc}^{S}$ and the components $C$ of $G[X\setminus S]$ such that $\rep{V_x}{2}(C)\in \cX_{\vc}^{\comp{S}}$.
		\end{enumerate}
	\end{definition}
	
	Similarly to Definition \ref{def:partialsolution}, we define the notion of \textit{complement solutions} associated with an index $i\in\bI_x$.
	We use this concept only to prove the correctness of our algorithm.
	
	\begin{definition}[Complement solutions]\label{def:complementsolution}
		Let $i=(\cX_{\vc}^{\comp{S}},\cX_{\vc}^{S}, \cX_{\comp{\vc}},\cY_{\vc}^{\comp{S}},\cY_{\vc}^{S})\in \bI_x$.
		We call complement solutions associated with $i$ all the pairs $(Y,\cP_Y)$ such that $Y\subseteq \comp{V_x}$, $\cP_Y$ is an $\comp{S}$-contraction of $Y$ and the following conditions are satisfied:
		\begin{enumerate}[(a)]
			\item for every $U\in \cY_{\vc}^{S}$, there exists a unique $v \in Y\cap S$ such that $U\equi{V_x}{1} \{v\}$,
			
			\item for every $U\in \cY_{\vc}^{\comp{S}}$, there exists a unique $P\in \cP_Y$ such that $U\equi{V_x}{2} P$,
			
			\item $G[Y]_{\downarrow\cP_Y}$ is a forest,
			
			\item for every $P\in\cP_Y$ and $\{v\}\in \cX_{\vc}^{S}$, we have $|N(v)\cap P |\leq 1$,
			
			\item for every $y\in Y\cap S$ and $R\in \cX_{\vc}^{\comp{S}}\cup \cP_Y$, we have $|N(y)\cap R| \leq 1$,
			
			\item $N(\cX_{\comp{\vc}})\cap V(\cY_{\comp{\vc}})=\emptyset$, where $\cY_{\comp{\vc}}$ contains the blocks $\{v\}\in \binom{Y\cap S}{1}$ such that $\rep{\comp{V_x}}{1}(\{v\})\notin \cY_{\vc}^{S}$ and the blocks $P\in\cP_Y$ such that $\rep{\comp{V_x}}{2}(P)\notin \cY_{\vc}^{\comp{S}}$.
		\end{enumerate}
	\end{definition}
	
	Let us give some explanations on the conditions of Definitions \ref{def:partialsolution} and \ref{def:complementsolution}.
	Let $X$ be a partial solution associated with an index $i\in \bI_x$ and $(Y,\cP_Y)$ be a complement solution associated with $i$.
	Conditions (a) and (b) of both definitions guarantee that there exists a subset $\VC$ of $X_{\downarrow \cc(X\setminus S)}\cup Y_{\downarrow\cP_Y}$  such that there is a one-to-one correspondence between the blocks of $\VC$ and the representatives in $\cX_{\vc}^{\comp{S}}\cup \cX_{\vc}^{S}\cup \cY_{\vc}^{\comp{S}}\cup \cY_{\vc}^{S}$.
	
	Condition (c) of Definition \ref{def:partialsolution} guarantees that the connections between $X_{\downarrow \cc(X\setminus S)}$ and $\VC$ are acyclic.
	As explained earlier, Conditions (d) and (e) of both definitions are here to control the $S$-cycles which might disappear with the vertex contractions.
	In particular, by Fact~\ref{fact:Sforestandforest}, Conditions~(c), (d) and (e) together imply that $G[X]$ and $G[Y]$ are $S$-forests.
	
	Finally, as explained earlier, the last conditions of both definitions ensure that $\VC$ the set described by $\cX_{\vc}^{\comp{S}},\cX_{\vc}^{S},\cY_{\vc}^{\comp{S}}$ and $\cY_{\vc}^{S}$ is a vertex cover of $G[X,Y]_{\downarrow \cc(X\setminus S)\cup \cP_Y}$.
	Notice that $X\setminus V(\VC_X)$ and $V(\cY_{\comp{\vc}})$ correspond the set of vertices in $X$ and $Y$, respectively, that do not belong to a block in the vertex cover $\VC$.
	Such observations are used to prove the following two results.
	
	\begin{lemma}\label{lemma:exitsAnIndex}
		Let $X\subseteq V_x$ and $Y\subseteq \comp{V_x}$ such that $G[X\cup Y]$ is an $S$-forest. 
		There exist $i\in \bI_x$ and an $\comp{S}$-contraction $\cP_Y$ of $Y$ such that (1)~$G[X\cup Y]_{\downarrow \cc(X\setminus S)\cup \cP_Y}$ is a forest, (2)~$X$ is a partial solution associated with $i$ and (3)~$(Y,\cP_Y)$ is a complement solution associated with $i$.
	\end{lemma}
	\begin{proof}
		By Lemma \ref{lem:equiStreeScontrac}, there exists an $\comp{S}$-contraction $\cP_Y$ of $Y$ such that the following properties are satisfied:
		\begin{enumerate}[(A)]
			\item $G[X\cup Y]_{\downarrow \cc(X\setminus S)\cup \cP_Y}$ is a forest,
			\item for all $P\in\cc(X\setminus S)\cup \cP_Y$ and all $v\in (X\cup Y)\cap S$, we have $|N(v)\cap P|\leq 1$,
			\item the graph $G[X, Y]_{\downarrow \cc(X\setminus S)\cup\cP_Y}$ admits a vertex cover $\VC$ of size at most $4\mim(V_x)$ such that the neighborhoods of the blocks in $\VC$ are pairwise distinct in $G[X, Y]_{\downarrow \cc(X\setminus S)\cup\cP_Y}$.
		\end{enumerate}
		We construct $i=(\cX_{\vc}^{\comp{S}},\cX_{\vc}^{S}, \cX_{\comp{\vc}},\cY_{\vc}^{\comp{S}},\cY_{\vc}^{S})\in\bI_x$ from $\VC$ as follows:
		\begin{itemize}
			\item $\cX_{\vc}^{S}=\{\rep{V_x}{1}(\{v\}) \mid \{v\}\in \binom{X\cap S}{1}\cap\VC\}$,
			\item $\cX_{\vc}^{\comp{S}}=\{\rep{V_x}{2}(P) \mid P\in \cc(X\setminus S) \cap \VC \}$,	
			\item $\cX_{\comp{\vc}}= \rep{V_x}{1}(X\setminus V(\VC))$,
			\item $\cY_{\vc}^{\comp{S}}= \{ \rep{\comp{V_x}}{2}(P)\mid P\in \cP_Y\cap \VC\}$,
			\item $\cY_{\vc}^{S}=\{\rep{\comp{V_x}}{1}(\{v\}) \mid \{v\}\in \binom{Y\cap S}{1}\cap \VC\}$.
		\end{itemize}
		Since $|\VC|\leq 4\mim(V_x)$, we have $|\cX_{\vc}^{\comp{S}}|+|\cX_{\vc}^{S}|+|\cY_{\vc}^{\comp{S}}| + |\cY_{\vc}^{S}| \leq 4\mim(V_x)$ and thus we have $i\in\bI_x$.
		\medskip
		
		We claim that $X$ is a partial solution associated with $i$	.
		By construction of $i$, Conditions (a), (b) and (f) of Definition~\ref{def:partialsolution} are satisfied.
		In particular, Condition (a) and (b) follow from Property~(C), i.e. the neighborhoods of the blocks in $\VC$ are pairwise distinct in $G[X, Y]_{\downarrow \cc(X\setminus S)\cup\cP_Y}$.
		So, the blocks in $X_{\downarrow \cc(X\setminus S)}\cap \VC$ are pairwise non-equivalent for the $d$-neighbor equivalence over $V_x$ for any $d\in \bN^+$ including $1$ and $2$.
		Consequently, there is a one to one correspondence between the blocks of $X_{\downarrow \cc(X\setminus S)}\cap \VC$ and the representatives in $\cX_{\vc}^{\comp{S}}\cup \cX_{\vc}^{S}$.
		
		It remains to prove Conditions (c), (d) and (e).
		We claim that Condition (c) is satisfied: $\aux(X,i)$ is a forest.
		Observe that, by construction, $\aux(X,i)$ is isomorphic to the graph $G[X_{\downarrow\cc(X\setminus S)} \mid  Y_{\downarrow \cP_Y}\cap \VC ]$.
		Indeed, for every $P\in Y_{\downarrow \cP_Y}\cap \VC$, by construction, there exists a unique $U\in \cY_{\vc}^{\comp{S}}\cup \cY_{\vc}^{S}$ such that $U\equi{\comp{V_x}}{1} P$ or $U\equi{\comp{V_x}}{2} P$.
		In both case, we have $N(U)\cap V_x=N(P)\cap V_x$ and thus, the neighborhood of $P$ in $G[X_{\downarrow\cc(X\setminus S)} \mid  Y_{\downarrow \cP_Y}\cap \VC ]$ is the same as the neighborhood of $U$ in $\aux(X,i)$.
		Since $G[X_{\downarrow\cc(X\setminus S)} \mid  Y_{\downarrow \cP_Y}\cap \VC ]$ is a subgraph of $G[X\cup Y]_{\downarrow \cc(X\setminus S)\cup \cP_Y}$ and this latter graph is a forest, we deduce that $\aux(X,i)$ is also a forest.
		
		We deduce that Conditions (d) and (e) are satisfied from property~(B) and Fact~\ref{fact:2neighborequivalence}.
		Consequently, $X$ is a partial solution associated with $i$.
		\medskip
		
		Let us now prove that $(Y,\cP_Y)$ is a complement solution associated with $i$.
		From the construction of $i$ and with the same argument used earlier, we deduce that that Conditions~(a) and (b) of Definition~\ref{def:complementsolution} are satisfied.
		By Property~(A) $G[X\cup Y]_{\downarrow \cc(X\setminus S)\cup \cP_Y}$ is a forest, as a subgraph, $G[Y]_{\downarrow \cP_Y}$ is also a forest and thus Condition~(c) is satisfied.
		Conditions (d) and (e) are satisfied from Property~(B) and Fact~\ref{fact:2neighborequivalence}.
		
		It remains to prove Condition (f):  $N(\cX_{\comp{\vc}})\cap \cY_{\comp{\vc}}=\emptyset$ where $\cY_{\comp{\vc}}$ contains the blocks $\{v\}\in \binom{Y\cap S}{1}$ such that $\rep{\comp{V_x}}{1}(\{v\})\notin \cY_{\vc}^{S}$ and the blocks $P\in\cP_Y$ such that $\rep{\comp{V_x}}{2}(P)\notin \cY_{\vc}^{\comp{S}}$.
		By construction, $\cY_{\comp{\vc}} = Y \setminus  V(\VC)$.
		Since, $\VC$ is a vertex cover of the graph $G[X, Y]_{\downarrow \cc(X\setminus S)\cup\cP_Y}$, there are no edges between $X_{\downarrow \cc(X\setminus S)}\setminus\VC$ and $Y_{\downarrow \cP_Y}\setminus \VC$ in $G[X, Y]_{\downarrow \cc(X\setminus S)\cup\cP_Y}$.
		We deduce that $N(X \setminus V(\VC)) \cap \cY_{\comp{\vc}}= \emptyset$.
		Since $X \setminus V(\VC)\equi{V_x}{1} \cX_{\comp{\vc}}$, we conclude that $N(\cX_{\comp{\vc}})\cap \cY_{\comp{\vc}}=\emptyset$.
		This proves that $(Y,\cP_Y)$ is a complement solution associated with $i$.
	\end{proof}
	
	\begin{lemma}\label{lemma:forestImpliesS-forest}
		Let $i=(\cX_{\vc}^{\comp{S}},\cX_{\vc}^{S}, \cX_{\comp{\vc}},\cY_{\vc}^{\comp{S}},\cY_{\vc}^{S})\in \bI_x$, $X$ be a partial solution associated with $i$ and $(Y,\cP_Y)$ be a complement solution associated with $i$. If the graph $G[X\cup Y]_{\downarrow \cc(X\setminus S)\cup \cP_Y}$ is a forest, then $G[X\cup Y]$ is an $S$-forest.
	\end{lemma}
	\begin{proof}
		Assume that $G[X\cup Y]_{\downarrow \cc(X\setminus S)\cup \cP_Y}$ is a forest.
		By Fact~\ref{fact:Sforestandforest}, in order to prove that $G[X\cup Y]$ is an $S$-forest, it is enough to prove that for all $v\in (X\cup Y)\cap S$ and all $P\in \cc(X\setminus S)\cup \cP_Y$, we have $|N(v)\cap P|\leq 1$.
		Let us prove this statement for a vertex $v\in X\cap S$ the proof is symmetric for $v\in Y\cap S$.
		Let $P\in(X\cup Y)_{\downarrow\cc(X\setminus S)\cup \cP_Y}$.
		If $P\notin (\cc(X\setminus S)\cup \cP_Y)$, then $P$ is a singleton in $\binom{X\cup Y}{1}$ and we are done.
		If $P\in \cc(X\setminus S)$, then Condition (e) of Definition \ref{def:partialsolution} guarantees that we have $|N(v)\cap P|\leq 1$.
	
		Assume now that $P\in \cP_Y$.
		Suppose first that $\rep{\comp{V_x}}{2}(P)\notin \cY_{\vc}^{\comp{S}}$.
		From Condition~(f) of Definition \ref{def:complementsolution}, we have $N(P)\cap \cX_{\comp{\vc}}=\emptyset$.
		Let $r=\rep{V_x}{1}(\{v\})$.
		From the definition of $\cX_{\comp{\vc}}$ in Definition~\ref{def:partialsolution}, we deduce that if $r\notin \cX_{\vc}^{S}$, then $N(v)\cap \comp{V_x}\subseteq N(\cX_{\comp{\vc}})$ and thus $N(v)\cap P=\emptyset$.
		On the other hand, if $r\in \cX_{\vc}^{S}$, then Condition~(d) of Definition~\ref{def:complementsolution} ensures that $|N(r)\cap P|\leq 1$ and thus $|N(v)\cap P|\leq 1$.
		Now, suppose that $\rep{\comp{V_x}}{2}(P)\in \cY_{\vc}^{\comp{S}}$.
		By Condition~(e) of Definition \ref{def:partialsolution}, we know that $|N(v)\cap \rep{\comp{V_x}}{2}(P)|\leq 1$.
		From Fact~\ref{fact:2neighborequivalence}, we conclude that $|N(v)\cap P|\leq 1$. This concludes the proof of Lemma~\ref{lemma:forestImpliesS-forest}.
	\end{proof}

	For each index $i\in\bI_x$, we will design an equivalence relation $\sim_i$ between the partial solutions associated with $i$.
	We will prove that, for any partial solutions $X$ and $W$ associated with $i$, if $X\sim_i W$, then, for any complement solution $Y\subseteq \comp{V_x}$ associated with $i$, the graph $G[X\cup Y]$ is an $S$-forest if and only if $G[W\cup Y]$ is an $S$-forest.
	Then, given a set of partial solutions $\cA$ whose size needs to be reduced, it is sufficient to keep, for each $i\in\bI_x$ and each equivalence class $\cC$ of $\sim_i$, one partial solution in $\cC$ of maximal weight.
	The resulting set of partial solutions has size bounded by $|\bI_x|\cdot (4\mim(V_x))^{4\mim(V_x)}$ because $\sim_i$ generates at most $(4\mim(V_x))^{4\mim(V_x)}$ equivalence classes.
	
	\medskip

	Intuitively, given two partial solutions $X$ and $W$ associated with $i=(\cX_{\vc}^{\comp{S}},\cX_{\vc}^{S}, \cX_{\comp{\vc}},\cY_{\vc}^{\comp{S}},\cY_{\vc}^{S})$, we have $X\sim_i W$ if the blocks of $\VC$ (i.e., the vertex cover described by $i 	$) are \textit{equivalently connected} in $G[X_{\downarrow \cc(X\setminus S)}\mid \cY_{\vc}^{\comp{S}} \cup  \cY_{\vc}^{S} ]$ and $G[W_{\downarrow \cc(W\setminus S)}\mid \cY_{\vc}^{\comp{S}} \cup  \cY_{\vc}^{S}  ]$.
	In order to compare these connections, we use the following notion.
	
	\begin{definition}[$\cc(X,i)$]\label{def:aux(X,i)}
		Let $i=(\cX_{\vc}^{\comp{S}},\cX_{\vc}^{S}, \cX_{\comp{\vc}},\cY_{\vc}^{\comp{S}},\cY_{\vc}^{S})\in \bI_x$ and $X\subseteq V_x$ be a partial solution associated with $i$. For each connected component $C$ of $\aux(X,i)$, we define the set $C_\vc$ as follows:
		\begin{itemize}
			\item for every $U\in C$ such that $U\in \cY_{\vc}^{\comp{S}} \cup  \cY_{\vc}^{S}$, we have $U \in C_\vc$,
			\item for every $\{v\}\in \binom{X\cap S}{1}\cap C$ such that $\{v\}\equi{V_x}{1} R$ for some $R\in \cX_{\vc}^{S}$, we have $R\in C_\vc$,
			\item for every $U\in \cc(X\setminus S)$ such that $U\equi{V_x}{2} R$ for some $R \in \cX_{\vc}^{\comp{S}}$, we have $R \in C_\vc$.
		\end{itemize}
		We define $\cc(X,i)$ as the collection $\{C_\vc \mid C \text{ is a connected component of } \aux(X,i)\}$.
	\end{definition}
	
	For a connected component $C$ of $\aux(X,i)$, the set $C_\vc$ contains $C \cap (\cY_{\vc}^{\comp{S}}\cup \cY_{\vc}^{S})$ and the representatives of the blocks in $C\cap X_{\downarrow \cc(X\setminus S)}\cap \VC$ with $\VC$ the vertex cover described by $i$.
	Consequently, for every $X\subseteq V_x$ and $i\in\bI_x$, the collection $\cc(X,i)$ is partition of $\cX_{\vc}^{\comp{S}}\cup \cX_{\vc}^{S}\cup \cY_{\vc}^{\comp{S}}\cup \cY_{\vc}^{S}$.
	For the example given in Figure \ref{fig:auxxi}, observe that $\cc(X,i)$ is the partition that contains $\{R_1,U_1,U_2\}$, $\{R_2,U_3\}$ and $\{U_4\}$ (see also Figure~\ref{fig:ccXi})
	
	\medskip
	
	Now we are ready to give the notion of equivalence between partial solutions.
	We say that two partial solutions $X,W$ associated with $i$ are $i$-equivalent, denoted by $X\sim_i W$, if $\cc(X,i)=\cc(W,i)$.
	Our next result is the most crucial step.
	As already explained, our task is to show equivalence between partial solutions under any complement solution with respect to $S$-forests.
	Figure~\ref{fig:equivalent} gives an example of two $i$-equivalent partial solutions.	
	
	\medskip
	
	\begin{lemma}\label{lemma:equivalence}
		Let $i=(\cX_{\vc}^{\comp{S}},\cX_{\vc}^{S}, \cX_{\comp{\vc}},\cY_{\vc}^{\comp{S}},\cY_{\vc}^{S})\in \bI_x$. For every partial solutions $X,W$ associated with $i$ such that $X\sim_i W$ and for every complement solution $(Y,\cP_Y)$ associated with $i$, the graph $G[X\cup Y]_{\downarrow \cc(X\setminus S)\cup \cP_Y}$ is a forest if and only if the graph $G[W\cup Y]_{\downarrow \cc(W\setminus S)\cup \cP_Y}$ is a forest.
	\end{lemma}
	\begin{proof}
		Let $X,W$ be two partial solutions associated with $i$ such that  $X\sim_i W$ and let $(Y,\cP_Y)$ be a complement solution associated with $i$.
		To prove this lemma, we show that if $G[W\cup Y]_{\downarrow \cc(W\setminus S)\cup \cP_Y}$ contains a cycle, then $G[X\cup Y]_{\downarrow \cc(X\setminus S)\cup \cP_Y}$ contains a cycle too. See Figure~\ref{fig:intuitionCycles} for some intuitions on this proof.
		We will use the following notation in this proof.
		
		For $Z\in \{X,W\}$, we denote by $\VC_Z$ the set that contains:
		\begin{itemize}
			\item all $\{v\}\in \binom{Z\cap S}{1}$ such that $\rep{V_x}{1}(\{v\})\in \cX_{\vc}^{S}$,
			\item all $P\in \cc(Z\setminus S)$ such that $\rep{V_x}{2}(P)\in \cX_{\vc}^{\comp{S}}$.
		\end{itemize}
		We define also $\VC_Y$ as the set that contains:
		\begin{itemize}
			\item all $\{v\}\in \binom{Y\cap S}{1}$ such that $\rep{\comp{V_x}}{1}(\{v\})\in \cY_{\vc}^{S}$,
			\item all $P\in \cP_Y$ such that $\rep{\comp{V_x}}{2}(P)\in \cY_{\vc}^{\comp{S}}$.
		\end{itemize}
		The sets $\VC_X,\VC_W$ and $\VC_Y$ contain the blocks in $X_{\downarrow \cc(X\setminus S)}, W_{\downarrow \cc(W\setminus S)}$ and $Y_{\downarrow \cP_Y}$, respectively, which belong to the vertex cover described by $i$.
		Finally, for each $Z\in\{X,W\}$, we define the following two edge-disjoint subgraphs of $G[Z\cup Y]_{\downarrow\cc(Z\setminus S)\cup\cP_Y}$:
		\begin{itemize}
			\item $G_Z=G[Z_{\downarrow \cc(Z\setminus S)} \mid \VC_Y]$,
			\item $\comp{G_Z}= G[Z\cup Y]_{\downarrow \cc(Z\setminus S)\cup\cP_Y} - G_Z$.
		\end{itemize}
		As explained in the proof of Lemma \ref{lemma:exitsAnIndex}, for any $Z\in \{X,W\}$, the graph $\aux(Z,i)$ is isomorphic to the graph $G_Z$.
		Informally, $G_Z$ contains the edges of $G[Z\cup Y]_{\downarrow \cc(Z\setminus S)\cup \cP_Y}$ which are induced by $Z_{\downarrow\cc(Z\setminus S)}$ and those between $Z_{\downarrow\cc(Z\setminus S)}$ and $\VC_Y$.
		The following fact implies that $\comp{G_Z}$ contains the edges of $G[Z\cup Y]_{\downarrow \cc(Z\setminus S)\cup \cP_Y}$ that are induced by $Y_{\downarrow \cP_Y}$ and those between $Y_{\downarrow \cP_Y} \setminus \VC_Y$ and $\VC_Z$.
		
		\begin{fact}\label{fact:vertexcover}
			For any $Z\in \{X,W\}$, the set $\VC_Z\cup\VC_Y$ is a vertex cover of $G[Z,Y]_{\downarrow \cc(Z\setminus S)\cup \cP_Y}$.
		\end{fact}
		\begin{proof}
			First observe that $N(Y\setminus V(\VC_Y)) \cap \cX_{\comp{\vc}}=\emptyset$  thanks to Condition~(f) of Definition \ref{def:complementsolution}.
			Moreover, we have $\cX_{\comp{\vc}}\equi{V_x}{1} Z\setminus V(\VC_Z)$ by Condition (f) of Definition \ref{def:partialsolution}.
			We conclude that there are no edges between $Y_{\downarrow\cP_Y}\setminus \VC_Y$ and $Z_{\downarrow \cc(Z\setminus S)}\setminus \VC_Z$ in $G[Z,Y]_{\downarrow \cc(Z\setminus S)\cup\cP_Y}$.
			Hence, $\VC_Z\cup \VC_Y$ is a vertex cover of $G[Z,Y]_{\downarrow \cc(Z\setminus S)\cup\cP_Y}$.
		\end{proof}
		
		\medskip
		
		Assume that $G[W\cup Y]_{\downarrow \cc(W\setminus S)\cup \cP_Y}$ contains a cycle $C$.
		Our task is to show that $G[X\cup Y]_{\downarrow \cc(X\setminus S)\cup \cP_Y}$ contains a cycle as well.
		We first explore properties of $C$ with respect to $G_W$ and $\comp{G_W}$.
		Since the graph $\aux(W,i)$ is a forest and it is isomorphic to $G_W$, we know that $C$ must contain at least one edge from $\comp{G_W}$.
		Moreover, $C$ must go though a block of $W_{\downarrow\cc(W\setminus S)}$ because $G[Y]_{\downarrow\cP_Y}$ is a forest.
		Consequently, (and because from Fact~\ref{fact:vertexcover} $\VC_W\cup \VC_Y$ is a vertex cover of $G[W\cup Y]_{\downarrow \cc(W\setminus S)\cup \cP_Y}$), we deduce that $C$ is the concatenation of edge-disjoint paths $P_1\dots P_t$ such that for each $\ell\in [t]$ we have:
		\begin{itemize}
			\item $P_\ell$ is a non-empty \textit{path} with endpoints in $\VC_W\cup \VC_Y$ and internal blocks not in $\VC_W\cup \VC_Y$ and $P_\ell$ is either a path of $G_W$ or $\comp{G_W}$.
		\end{itemize}
		At least one of these paths is in $\comp{G_W}$ and, potentially, $C$ may be entirely contained in $\comp{G_W}$.
		Figure~\ref{fig:cyclewy} presents two possible interactions between $C$ and the graphs $G_W$ and $\comp{G_W}$.
		
		\begin{figure}[h]
			\centering
			\includegraphics[width=0.7\linewidth]{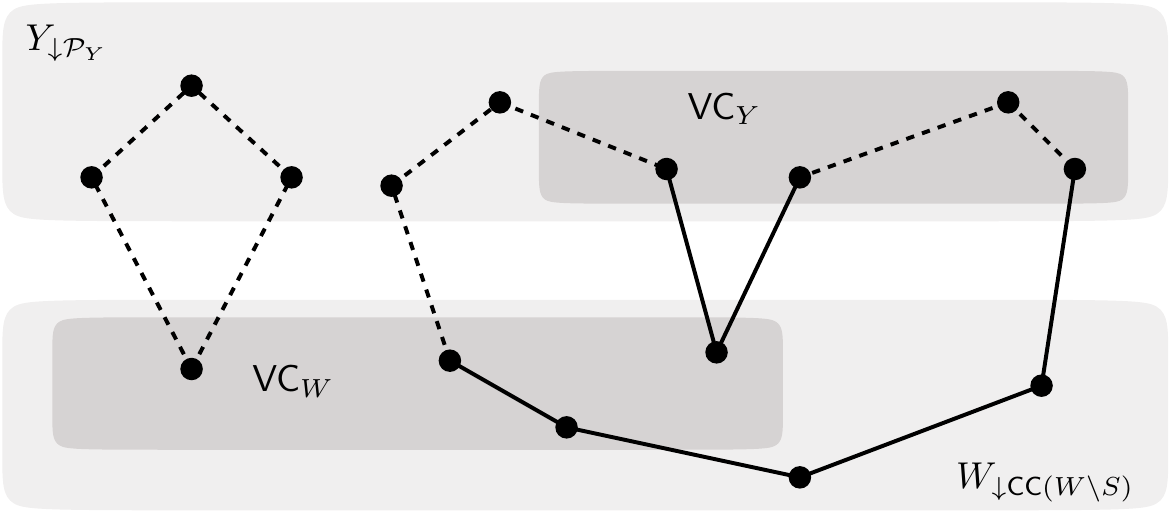}
			\caption{How cycles in $G[W\cup Y]_{\downarrow \cc(W\setminus S)\cup \cP_Y}$ may interact with the graphs $G_W$ and $\comp{G_W}$.
				The solid edges belong to $G_W$ and the dashed edges belong to $\comp{G_W}$.}
			\label{fig:cyclewy}
		\end{figure}
		
		Given an endpoint $U\in \VC_W\cup \VC_Y$ of one of the paths $P_1,\dots,P_\ell$, we define $U_X$ and $U_i$ as the analogs of $U$ in $\VC_X\cup \VC_Y$ and $\cX_{\vc}^{\comp{S}}\cup \cX_{\vc}^{S} \cup \cY_{\vc}^{\comp{S}}\cup  \cY_{\vc}^{S}$, respectively, as follows:
		\begin{itemize}
			\item if $U\in \cc(W\setminus S)$, then $U_X$ and $U_i$ are the unique elements of $\cc(X\setminus S)$ and $\cX_{\vc}^{\comp{S}}$, respectively, such that $U\equi{V_x}{2} U_X \equi{V_x}{2} U_i$,
			
			\item if $U=\{v\}\in \binom{W\cap S}{1}$, then $U_X$ and $U_i$ are the unique elements in $\binom{X\cap S}{1}$ and $\cX_{\vc}^{S}$, respectively, such that $U\equi{V_x}{1} U_X \equi{V_x}{1} U_i$,
			
			\item if $U\in \cP_Y$, then $U_X=U$ and $U_i$ is the unique element of $\cY_{\vc}^{\comp{S}}$ such that $U\equi{\comp{V_x}}{2} U_i$;
			
			\item otherwise, if $U=\{v\}\in \binom{Y\cap S}{1}$, then $U_X=U$ and $U_i$ is the unique element of $\cY_{\vc}^{S}$ such that $U\equi{\comp{V_x}}{1} U_i$.
		\end{itemize}
		Observe that $U_X$ and $U_i$ exist by Conditions (a) and (b) of Definition \ref{def:partialsolution} and Definition \ref{def:complementsolution}.
		
		For each $\ell\in [t]$, we construct a non-empty path $P'_\ell$ whose endpoints are the analogs in $\VC_X\cup \VC_Y$ of the endpoints of $P_\ell$ and such that if $P_\ell$ is a path in $G_W$ (resp. $\comp{G_W}$), then $P'_\ell$ is a path in $G_X$ (resp. $\comp{G_X}$).
		This is sufficient to prove the claim because thanks to this, we can construct a closed walk in $G[X\cup Y]_{\downarrow\cc(X\setminus S)\cup\cP_Y}$ by concatenating the paths $P'_1,\dots,P'_t$.
		Since $G_X$ and $\comp{G_X}$ are edge-disjoint and the paths $P'_1,\dots,P'_t$ are non-empty paths, this closed walk must contain a cycle.
		
		\medskip
		
		Let $\ell\in [t]$ and $U,T$ be the endpoints of $P_\ell$.
		We denote by $U_X,T_X,U_i$ and $T_i$ the analogs of $U$ and $T$ in $\VC_X\cup \VC_Y$ and $\cX_{\vc}^{\comp{S}}\cup \cX_{\vc}^{S} \cup \cY_{\vc}^{\comp{S}}\cup  \cY_{\vc}^{S}$, respectively.
		
		First, assume that $P_\ell$ is path of $G_W$.
		Observe that $U_i$ and $T_i$ belong to the same partition class of $\cc(W,i)$.
		This follows from the definitions of $U_i,T_i$ and the fact that $G_W$ is isomorphic to $\aux(W,i)$.
		As $W\sim_i X$, we deduce that $U_i$ and $T_i$ belong to the same partition class of $\cc(X,i)$.
		By construction, $U_i$ and $T_i$ are the analogs of $U_X$ and $T_X$ in $\cX_{\vc}^{\comp{S}}\cup \cX_{\vc}^{S} \cup \cY_{\vc}^{\comp{S}}\cup  \cY_{\vc}^{S}$.
		We conclude that $U_X$ and $T_X$ are connected in $G_X$ via a path $P'_\ell$.
		We claim that $P'_\ell$ is not empty because we have $U_X\neq T_X$.
		As $P_\ell$ is a non empty path of $G_W$ and because $G_W$ is acyclic, we know that $U$ and $T$ are distinct.
		Hence, by the construction of $U_X$ and $T_X$, we deduce that $U_X\neq T_X$.
		
		\medskip
		
		Now, assume that $P_\ell$ is a non-empty path of $\comp{G_W}$.
		Since $\VC_W\cup \VC_Y$ is a vertex cover of $G[W,Y]_{\downarrow \cc(W\setminus S)\cup \cP_Y}$, the blocks in $W_{\downarrow \cc(W\setminus S)}$ that do not belong to $VC_W$ are isolated in $\comp{G_W}$.
		As $P_\ell$ is not empty, the blocks of $P_\ell$ which belong to $W_{\downarrow \cc(W\setminus S)}$ are in $\VC_W$.
		Because the internal blocks of the paths $P_1,\dots,P_t$ are not in $\VC_W\cup \VC_Y$, we deduce that the internal blocks of $P_\ell$ belong to $Y_{\downarrow \cP_Y}$. We distinguish the following cases:
		\begin{itemize}
			\item If both endpoints of $P_\ell$ belong to $\VC_Y$, then $U_X=U$, $T_X=T$ and all the blocks of $P_\ell$ belong to $Y_{\downarrow \cP_Y}$. It follows that $P_\ell$ is a non-empty path of $\comp{G_X}$ because $G[Y]_{\downarrow \cP_Y}$ is a subgraph of $\comp{G_X}$. In this case, we take $P'_\ell=P_\ell$.
			
			\item Assume now that one or two endpoints of $P_\ell$ belong to $\VC_W$. 
			Suppose w.l.o.g. that $U$ belongs to $\VC_W$.
			Since $P_\ell$ is non-empty and the internal blocks of $P_\ell$ are in $Y_{\downarrow \cP_Y}$, $U$ has a neighbor $Q\in Y_{\downarrow \cP_Y}$ in $P_\ell$.
			We claim that $Q$ is adjacent to $U_X$ in $\comp{G_X}$.
			By definition of $U_X$, we have $U\equi{V_x}{d} U_X$ for some $d\in\{1,2\}$ and in particular $N(U)\cap \comp{V_x}= N(U_X)\cap \comp{V_x}$.
			As $U$ and $Q$ are adjacent in $\comp{G_W}$, we deduce that $N(U)\cap Q \neq \emptyset$.
			It follows that $N(U_X)\cap Q\neq\emptyset$ and thus $Q$ and $U_X$ are adjacent in $\comp{G_X}$.
			Symmetrically, we can prove that if $T\in \VC_W$, then the neighbor of $T$ in $P_\ell$ is adjacent to $T_X$ in $\comp{G_X}$.
			
			Hence, the neighbors of $U$ and $T$ in $P_\ell$ are adjacent to $U_X$ and $T_X$ respectively in $\comp{G_X}$.
			We obtain $P_\ell'$ from $P_\ell$ by replacing $U$ and $T$ by $U_X$ and $T_X$.
			Since the internal blocks of $P_\ell$ belong to $Y_{\downarrow \cP_Y}$ and $G[Y]_{\downarrow \cP_Y}$ is a subgraph of $\comp{G_X}$, we deduce that $P'_\ell$ is a path of $\comp{G_X}$.
			The path $P'_\ell$ is not-empty because it contains $U_X$ and $Q$ which are distinct blocks of $\comp{G_X}$ as $U_X\in \VC_X$ (since $U\in\VC_W$ by assumption) and $Q\in Y_{\downarrow \cP_Y}$.
		\end{itemize}
	\end{proof}
	
	
	\medskip
	
	The following theorem proves that, for every set of partial solutions $\cA\subseteq 2^{V_x}$, we can compute a small subset $\cB\subseteq \cA$ such that $\cB$ \textit{represents} $\cA$, i.e., for every $Y\subseteq \comp{V_x}$, the best solutions we obtain from the union of $Y$ with a set in $\cA$ are as good as the ones we obtain from $\cB$.
	Firstly, we formalize this notion of representativity.
	
	\begin{definition}[Representativity]\label{def:representativity}
		For every $\cA\subseteq 2^{V_x}$ and $Y\subseteq \comp{V_x}$, we define
		\[ \best(\cA,Y)= \max\{ \w(X) \mid X\in \cA \text{ and } G[X\cup Y] \text{ is an } S\text{-forest} \} .\]
		Given $\cA,\cB\subseteq 2^{V_x}$, we say that $\cB$ \emph{represents} $\cA$ if, for every $Y\subseteq \comp{V_x}$, we have $\best(\cA,Y)=\best(\cB,Y)$.
	\end{definition}
	
	We recall that $\snec_2(A)=\max(\nec_2(A),\nec_2(\comp{A}))$.
		
	\begin{lemma}\label{lem:reduce}
		There exists an algorithm $\reduce$ that, given a set $\cA\subseteq 2^{V_x}$, outputs in time $O(|\cA|\cdot |\bI_x| \cdot (4\mim(V_x))^{4\mim(V_x)}\cdot  \log(\snec_2(V_x)) \cdot n^3)$ a subset $\cB\subseteq \cA$ such that $\cB$ represents $\cA$ and $|\cB|\leq |\bI_x| \cdot (4\mim(V_x))^{4\mim(V_x)}$.
	\end{lemma}
	
	\begin{proof}
		Given $\cA\subseteq 2^{V_x}$ and $i\in\bI_x$, we define $\reduce(\cA,i)$ as the operation which returns a set containing one partial solution $X\in \cA$ associated with $i$ of each equivalence class of $\sim_i$ such that $\w(X)$ is maximum.
		Moreover, we define $\reduce(\cA)= \bigcup_{i\in\bI_x}\reduce(\cA,i)$.
		
		\medskip
		
		We prove first that $\reduce(\cA)$ represents $\cA$, that is $\best(\cA,Y)=\best(\reduce(\cA),Y)$ for all $Y\subseteq \comp{V_x}$.
		Let $Y\subseteq \comp{V_x}$.
		Since $\reduce(\cA)\subseteq \cA$, we already have $\best(\reduce(\cA),Y)\leq \best(\cA,Y)$.
		Consequently, if there is no $X\in\cA$ such that $G[X\cup Y]$ is an $S$-forest, we have $\best(\reduce(\cA),Y)= \best(\cA,Y)=\max(\emptyset)=-\infty$.
		
		Assume that there exists $X\in\cA$ such that $G[X\cup Y]$ is an $S$-forest.
		Let $X\in \cA$ such that $G[X\cup Y]$ is an $S$-forest and $\w(X)=\best(\cA,Y)$.
		By Lemma \ref{lemma:exitsAnIndex}, there exists $i\in\bI_x$ and an $\comp{S}$-contraction $\cP_Y$ of $Y$ such that (1)~$G[X\cup Y]_{\downarrow \cc(X\setminus S)\cup \cP_Y}$ is a forest, (2)~$X$ is a partial solution associated with $i$ and (3)~$(Y,\cP_Y)$ is a complement solution associated with $i$.
		
		From the construction of $\reduce(\cA,i)$, there exists $W\in \reduce(\cA)$ such that $W$ is a partial solution associated with $i$, $X\sim_i W$, and $\w(W)\geq \w(X)$.
		By Lemma \ref{lemma:equivalence} and since $G[X\cup Y]_{\downarrow \cc(X\setminus S)\cup \cP_Y}$ is a forest, we deduce that $G[W\cup Y]_{\downarrow \cc(W\setminus S)\cup \cP_Y}$ is a forest too.
		Thanks to Lemma \ref{lemma:forestImpliesS-forest}, we deduce that $G[W\cup Y]$ is an $S$-forest.
		As $\w(W)\geq \w(X)=\best(\cA,Y)$, we conclude that $\best(\cA,Y)=\best(\reduce(\cA),Y)$.
		Hence, $\reduce(\cA)$ represents $\cA$.
		
		\medskip
		
		We claim that $|\reduce(\cA)| \leq  |\bI_x| \cdot (4\mim(V_x))^{4\mim(V_x)}$.
		For every $i=(\cX_{\vc}^{\comp{S}},\cX_{\vc}^{S}, \cX_{\comp{\vc}},\cY_{\vc}^{\comp{S}},\cY_{\vc}^{S})\in\bI_x$ and  partial solution $X$ associated with $i$, $\cc(X,i)$ is a partition of $\cX_{\vc}^{\comp{S}}\cup \cX_{\vc}^{S}\cup  \cX_{\comp{\vc}} \cup \cY_{\vc}^{\comp{S}} \cup \cY_{\vc}^{S}$.
		Since $|\cX_{\vc}^{\comp{S}}|+|\cX_{\vc}^{S}|+|\cY_{\vc}^{\comp{S}}| + |\cY_{\vc}^{S}| \leq 4\mim(V_x)$, there are at most $(4\mim(V_x))^{4\mim(V_x)}$ possible values for $\cc(X,i)$.
		We deduce that, for every $i\in\bI_x$, the relation $\sim_i$ generates at most $(4\mim(V_x))^{4\mim(V_x)}$ equivalence classes, so $|\reduce(\cA,i)|\leq (4\mim(V_x))^{4\mim(V_x)}$ for every $i\in \bI_x$.
		By construction, we conclude that $|\reduce(\cA)|\leq |\bI_x| \cdot (4\mim(V_x))^{4\mim(V_x)}$.
		
		\medskip
		
		It remains to prove the runtime.
		As $\nec_1(V_x)\leq \nec_2(V_x)$, by Lemma \ref{lem:computenecd} we can  compute in time $O(\snec_2(V_x) \cdot \log(\snec_2(V_x)) \cdot n^2)$ the sets  $\Rep{V_x}{1},\Rep{V_x}{2}$, $\Rep{\comp{V_x}}{2}$ and  data structures which compute $\rep{V_x}{1},\rep{V_x}{2}$ and $\rep{\comp{V_x}}{2}$ in time $O(\log(\snec_2(V_x))\cdot n^2)$.
		Given $\Rep{V_x}{1},\Rep{V_x}{2}$, $\Rep{\comp{V_x}}{2}$, we can compute $\bI_x$ in time $O(|\bI_x|\cdot n^2)$.
		Since, $\snec_2(V_x)\leq |\bI_x|$, the time required to compute these sets and data structures is less than $O(|\cA|\cdot |\bI_x| \cdot (4\mim(V_x))^{4\mim(V_x)}\cdot  \log(\snec_2(V_x)) \cdot n^3)$.
		
		For each $i\in\bI_x$ and $X\in \cA$, we can decide whether $X$ is a partial solution associated with $i$ and compute $\aux(X,i)$, $\cc(X,i)$ in time $O(\log(\snec_2(V_x))\cdot n^3)$.
		For doing so, we simply start by computing $\rep{V_x}{2}(C)$ and $\rep{V_x}{1}(\{v\})$ for each $C\in \cc(X\setminus S)$ and $\{v\}\in \binom{X\cap S}{1}$, this is doable in $O( \log(\snec_2(V_x))\cdot n^3)$ since $|\cc(X\setminus S)|+|X\cap S|\leq n$. Then with standard algorithmic techniques, we check whether $X$ satisfies all the conditions of Definition~\ref{def:partialsolution} and compute $\aux(X,i)$ and $\cc(X,i)$.
		
		Given two partial solutions $X,W$ associated with $i$, $\cc(X,i)$ and $\cc(W,i)$, we can decide whether $X\sim_i W$ in time $O(\mim(V_x)) \leq O(n)$.
		We deduce that, for each $i\in\bI_x$, we can compute $\reduce(\cA,i)$ in time
		$O(|\cA|\cdot (4\mim(V_x))^{4\mim(V_x)} \cdot \log(\snec_2(V_x))\cdot n^3)$.
		We deduce the running time to compute $\reduce(\cA)$ by multiplying the running time needed for $\reduce(\cA,i)$ by $|\bI_x|$.
	\end{proof}
	
	We are now ready to prove the main theorem of this paper. For two subsets $\cA$ and $\cB$	of $2^{V(G)}$, we define the \emph{merging} of $\cA$ and $\cB$, denoted by $\cA\otimes \cB$, as 	$\cA \otimes \cB :=\{ X\cup Y \mid X\in \cA \,\text{ and }\, Y\in \cB\}$.
	Observe that $\cA\otimes \cB=\emptyset$ whenever $\cA=\emptyset$ or $\cB=\emptyset$.
	
	\begin{theorem}\label{thm:main}
		There exists an algorithm that, given an $n$-vertex graph $G$ and a rooted layout $(T,\delta)$ of $G$, solves \textsc{Subset Feedback Vertex Set} in time
		\[ O\left(\sum_{x\in V(T)} |\bI_x|^3\cdot (4\mim(V_x))^{12\mim(V_x)}\cdot \log(\snec_2(V_x)) \cdot n^3\right). \]
	\end{theorem}
	\begin{proof}
		The algorithm is a usual bottom-up dynamic programming algorithm.
		For every node $x$ of $T$, the algorithm computes a set of partial solutions $\cA_x \subseteq 2^{V_x}$  such that $\cA_x$ represents $2^{V_x}$ and $|\cA_x|\leq |\bI_x| \cdot (4\mim(V_x))^{4\mim(V_x)}$.
		For the leaves $x$ of $T$ such that $V_x=\{v\}$, we simply take $\cA_x= 2^{V_x} = \{\emptyset, \{v\}\}$.
		In order to compute $\cA_x$ for $x$ an internal node of $T$ with $a$ and $b$ as children, our algorithm will simply compute $\cA_x= \reduce(\cA_a\otimes \cA_b)$.
		Once the the set $\cA_r$ is computed with $r$ the root of $T$, our algorithm outputs a set $X\in \cA_r$ of maximum weight.
		
		\medskip
		
		By Lemma~\ref{lem:reduce}, we have $|\cA_x|\leq |\bI_x|\cdot (4\mim(V_x))^{4\mim(V_x)}$, for every node $x$ of $T$.
		The following claim helps us to prove that $\cA_x$ represents $2^{V_x}$ for the internal nodes $x$ of $T$.
		
		\begin{claim}\label{claim:representsPropagation}
			Let $x$ be an internal of $T$ with $a$ and $b$ as children.
			If $\cA_a$ and $\cA_b$ represent, respectively, $2^{V_a}$ and $2^{V_b}$, then $\reduce(\cA_a\otimes \cA_b)$ represents $2^{V_x}$.
		\end{claim}
		\begin{proof}
			Assume that  $\cA_a$ and $\cA_b$ represent, respectively, $2^{V_a}$ and $2^{V_b}$.
			First, we prove that $\cA_a\otimes \cA_b$ represents $2^{V_x}$.
			We have to prove that, for every $Y\subseteq \comp{V_x}$, we have $\best(\cA_a\otimes \cA_b,Y)=\best(2^{V_x},Y)$.
			Let $Y\subseteq \comp{V_x}$.
			By definition of $\best$, we have the following
			\begin{align*}
			\best(\cA_a\otimes \cA_b,Y)	&= \max\{ \w(X)+ \w(W)\mid X\in \cA_a\wedge W\in \cA_b\wedge G[X\cup W\cup Y] \text{ is an $S$-forest} \} \\
			&= \max\{ \best(\cA_a,W\cup Y) + \w(W) \mid  W\in \cA_b \}.
			\end{align*}
			As $\cA_a$ represents $2^{V_a}$, we have $\best(\cA_a,W\cup Y)=\best(2^{V_a},W\cup Y)$ and we deduce that $\best(\cA_a\otimes \cA_b,Y)=\best(2^{V_a}\otimes\cA_b,Y)$.
			Symmetrically, as $\cA_b$ represents $2^{V_b}$, we infer that $\best(2^{V_a}\otimes\cA_b,Y)=\best(2^{V_a}\otimes 2^{V_b},Y)$.
			Since $2^{V_a}\otimes 2^{V_b}= 2^{V_x}$, we conclude that $\best(\cA_a\otimes \cA_b,Y)$ equals $\best(2^{V_x},Y)$.
			As this holds for every $Y$, it proves that $\cA_a\otimes \cA_b$ represents $2^{V_x}$.
			By Lemma~\ref{lem:reduce}, we know that $\reduce(\cA_a\otimes \cA_b)$ represents $\cA_a\otimes \cA_b$.
			As the relation ``represents'' is transitive, we conclude that $\reduce(\cA_a\otimes \cA_b)$ represents $2^{V_x}$.
		\end{proof}
		
		For the leaves $x$ of $T$, we obviously have that $\cA_x$ represents $2^{V_x}$, since $\cA_x=2^{V_x}$.
		From Claim~\ref{claim:representsPropagation} and by induction, we deduce that $\cA_x$ represents $2^{V_x}$ for every node $x$ of $T$.
		In particular, $\cA_r$ represents $2^{V(G)}$ with $r$ the root of $T$.
		By Definition \ref{def:representativity}, $\cA_r$ contains a set $X$ of maximum size such that $G[X]$ is an $S$-forest.
		This proves the correctness of our algorithm.
		
		\medskip
		
		It remains to prove the running time.
		Observe that, for every internal node $x$ of $T$ with $a$ and $b$ as children, the size of $\cA_a\otimes \cA_b$ is at most $|\bI_x|^2\cdot (4\mim(V_x))^{8\mim(V_x)}$ and it can be computed in time $O(|\bI_x|^2\cdot (4\mim(V_x))^{8\mim(V_x)}\cdot n^2)$.
		By Lemma~\ref{lem:reduce}, the set $\cA_x=\reduce(\cA_a\otimes \cA_b)$ is computable in time $O(|\bI_x|^3\cdot (4\mim(V_x))^{12\mim(V_x)}\cdot \log(\snec_2(V_x))\cdot n^3)$.
		This proves the running time.
	\end{proof}
	
	\subsection{Algorithmic consequences}
	
	In order to obtain the algorithmic consequences of our meta-algorithm given in Theorem~\ref{thm:main}, we need the following lemma
	which bounds the size of each set of indices with respect to the considered parameters.
	
	\begin{lemma}\label{lem:consequences}
		For every $x\in V(T)$, the size of $\bI_x$ is upper bounded by:
		\begin{multicols}{3}
			\begin{itemize}
				\item $2^{O(\rw(V_x)^3)}$,
				\item $2^{O(\Qrw(V_x)^2\log(\Qrw(V_x)))}$,
				\item $n^{O(\mim(V_x)^2)}$.
			\end{itemize}
		\end{multicols}
	\end{lemma}
	\begin{proof}
		For $A\subseteq V(G)$, let $\mw(A)$ be the number of different rows in the matrix $M_{A,\comp{A}}$.
		Observe that, for every $A\subseteq V(G)$, we have $\mw(A)=\{ \rep{A}{1}(\{v\}) \mid v\in A\}$.
		Remember that the $\cX_{\vc}^{S}$'s and $\cY_{\vc}^{S}$'s  are subsets of $\{ \rep{V_x}{1}(\{v\}) \mid v\in V_x\}$ and $\{ \rep{\comp{V_x}}{1}(\{v\}) \mid v\in \comp{V_x}\}$ respectively.
		From Definition \ref{def:indices}, we have $|\cX_{\vc}^{\comp{S}}|+|\cX_{\vc}^{S}|+|\cY_{\vc}^{\comp{S}}| + |\cY_{\vc}^{S}| \leq 4\mim(V_x)$, for every $(\cX_{\vc}^{\comp{S}},\cX_{\vc}^{S}, \cX_{\comp{\vc}},\cY_{\vc}^{\comp{S}},\cY_{\vc}^{S})\in\bI_x$.
		Thus, the size of $\bI_x$ is at most
		\[ \nec_1(V_x)\cdot \big(\nec_2(V_x)+ \mw(V_x) + \nec_2(\comp{V_x}) + \mw(\comp{V_x})\big)^{4\mim(V_x)}. \]
		
		\paragraph{\bf Rank-width.} By Lemma \ref{lem:compare}, we have $\nec_1(V_x)\leq 2^{\rw(V_x)^2}$ and $\nec_2(V_x),\nec_2(\comp{V_x})\leq 2^{2\rw(V_x)^2}$.
		Moreover, there is at most $2^{\rw(V_x)}$ different rows in the matrices $M_{V_x,\comp{V_x}}$ and  $M_{\comp{V_x},V_x}$, so $\mw(V_x)$ and $\mw(\comp{V_x})$ are upper bounded by $2^{\rw(V_x)}$.
		By Lemma \ref{lem:comparemim}, we have $4\mim(V_x)\leq 4\rw(V_x)$.
		We deduce from these inequalities that $|\bI_x|\leq 2^{\rw(V_x)^2}\cdot (2^{2\rw(V_x)^2 + 1} + 2^{\rw(V_x)+1})^{4\rw(V_x)}\in 2^{O(\rw(V_x)^3)}$.
		
		\paragraph{\bf $\bQ$-rank-width.} By Lemma \ref{lem:compare}, we have $\nec_1(V_x),\nec_2(V_x),\nec_2(\comp{V_x})\in 2^{O(\Qrw(V_x)\log(\Qrw(V_x)))}$.
		Moreover, there is at most $2^{\Qrw(V_x)}$ different rows in the matrices $M_{V_x,\comp{V_x}}$ and  $M_{\comp{V_x},V_x}$, so $\mw(V_x)$ and $\mw(\comp{V_x})$ are upper bounded by $2^{\Qrw(V_x)}$.
		By Lemma \ref{lem:comparemim}, we have $4\mim(V_x)\leq 4\Qrw(V_x)$.
		We deduce from these inequalities that
		\[ 	|\bI_x|\leq 2^{O(\Qrw(V_x)\log(\Qrw(V_x)))}\cdot \Big(2^{O(\Qrw(V_x)\log(\Qrw(V_x)))} + 2^{\Qrw(V_x)+1}\Big)^{4\Qrw(V_x)}. \]
		We conclude that $|\bI_x|\in 2^{O(\Qrw(V_x)^2\log(\Qrw(V_x)))}$.
		
		\paragraph{\bf Mim-width.} By Lemma \ref{lem:compare}, we know that $\nec_1(V_x)\leq |V_x|^{\mim(V_x)}$, $\nec_2(V_x) \leq |V_x|^{2\mim(V_x)}$, and $\nec_2(\comp{V_x})\leq |\comp{V_x}|^{2\mim(V_x)}$.
		We can assume that $n>2$ (otherwise the problem is trivial), so $\nec_2(V_x)+\nec_2(\comp{V_x})\leq |V_x|^{2\mim(V_x)} + |\comp{V_x}|^{2\mim(V_x)} \leq n^{2\mim(V_x)}$.
		Moreover, notice that, for every $A\subseteq V(G)$, we have $\{ \rep{A}{1}(\{v\}) \mid v\in A\}\leq |A|$.
		
		We deduce that $|\bI_x|\leq n^{\mim(V_x)} \cdot (n+n^{2\mim(V_x)})^{4\mim(V_x)}$.
		As we assume that $n>2$, we have $|\bI_x|\leq n^{8\mim(V_x)^2+5\mim(V_x)}\in n^{O(\mim(V_x)^2)}$.
		
	\end{proof}
	
	Now we are ready to state our algorithms with respect to the parameters rank-width $\rw(G)$ and $\bQ$-rank-width $\Qrw(G)$.
	In particular, with our next result we show that \textsc{Subset Feedback Vertex Set} is in \FPT parameterized by $\Qrw(G)$ or $\rw(G)$.
	
	
	\begin{theorem}\label{thm:qrankmain}
		There exist algorithms that solve \textsc{Subset Feedback Vertex Set} in time $2^{O(\rw(G)^3)}\cdot n^4$ and $2^{O(\Qrw(G)^2\log(\Qrw(G))))}\cdot n^{4}$.
	\end{theorem}
	\begin{proof}
		We first compute a rooted layout $\cL=(T,\delta)$ of $G$ such that $\rw(\cL)\in O(\rw(G))$ or $\Qrw(\cL)\in O(\Qrw(G))$.
		This is achieved through a $(3k + 1)$-approximation algorithm that runs in \FPT time $O(8^k \cdot n^4)$ parameterized by $k \in \{\rw(G),\Qrw(G)\}$ \cite{Oum09a}.
		Then, we apply the algorithm given in Theorem~\ref{thm:main}.
		Observe that for every node $x\in V(T)$, by Lemma~\ref{lem:consequences},
		$|\bI_x|^3$ lies in  $2^{O(\rw(V_x)^3)}$ and $2^{O(\Qrw(V_x)^2\log(\Qrw(V_x)))}$ and by Lemma~\ref{lem:compare}, $\snec_2(V_x)$ lies in $2^{O(\rw(V_x)^2)}$ and $2^{O(\Qrw(V_x)\log(\Qrw(V_x)))}$.
		Moreover, Lemma~\ref{lem:comparemim} implies that $\mim(V_x)^{\mim(V_x)}$ is upper bounded by $2^{\rw(G) \log(\rw(G))}$ and $2^{\Qrw(G)\log(\Qrw(G))}$.
		Therefore, we get the claimed runtimes for SFVS since $T$ contains $2n-1$ nodes.
	\end{proof}
	
	Regarding mim-width, our algorithm given below shows that \textsc{Subset Feedback Vertex Set} is in \XP parameterized by the mim-width of a given rooted layout.
	Note that we cannot solve SFVS in \FPT time parameterized by the mim-width of a given rooted layout unless $\FPT=\W[1]$,
	since \textsc{Subset Feedback Vertex Set} is known to be \W[1]-hard for this parameter even for the special case of $S=V(G)$ \cite{JaffkeKT20}.
	Moreover, contrary to the algorithms given in Theorem~\ref{thm:qrankmain}, here we need to assume that the input graph is given with a rooted layout.
	However, our next result actually provides a unified polynomial-time algorithm for \textsc{Subset Feedback Vertex Set}
	on well-known graph classes having bounded mim-width and for which a layout of bounded mim-width can be computed in polynomial time \cite{BelmonteV13}
	(e.g., \textsc{Interval} graphs, \textsc{Permutation} graphs, \textsc{Circular Arc}  graphs, \textsc{Convex} graphs, \textsc{$k$-Polygon},
	\textsc{Dilworth-$k$} and \textsc{Co-$k$-Degenerate} graphs for fixed $k$).
	
	\begin{theorem}\label{thm:mimmain}
		There exists an algorithm that, given an $n$-vertex graph $G$ and a rooted layout $\cL$ of $G$, solves \textsc{Subset Feedback Vertex Set} in time $n^{O(\mim(\cL)^2)}$.
	\end{theorem}
	\begin{proof}
		We apply the algorithm given in Theorem~\ref{thm:main}.
		By Lemma~\ref{lem:compare}, we have $\snec_2(V_x)\leq n^{O(\mim)}$ and from Lemma~\ref{lem:consequences} we have $|\bI_x|^3\in n^{O(\mim(V_x)^2)}$.
		The claimed runtime for SFVS follows by the fact that the rooted tree $T$ of $\cL$ contains $2n-1$ nodes.
	\end{proof}

	Let us relate our results for \textsc{Subset Feedback Vertex Set} to the \textsc{Node Multiway Cut}.
	It is known that \textsc{Node Multiway Cut} reduces to \textsc{Subset Feedback Vertex Set} \cite{FominHKPV14}.
	In fact, we can solve \textsc{Node Multiway Cut} by adding a new vertex $v$ with a large weight that is adjacent to all terminals and, then, run our algorithms for \textsc{Subset Feedback Vertex Set} with $S= \{v\}$ on the resulting graph.
	Now observe that any extension of a rooted layout $\cL$ of the original graph to the resulting graph has mim-width $\mim(\cL)+1$.
	Therefore, all of our algorithms given in Theorems~\ref{thm:qrankmain} and \ref{thm:mimmain} have the same running times for the \textsc{Node Multiway Cut} problem.
	
	\section{Conclusion }

This paper highlights the importance of the $d$-neighbor-equivalence relation to obtain meta-algorithm for several width measures at once.
	We extend the range of applications of this relation~\cite{BergougnouxK19esa,BuiXuanTV13,GolovachHKKSV18,OumSV13}
by proving that it is useful for the atypical acyclicity constraint of the \textsc{Subset Feedback Vertex Set} problem.
	It would be interesting to see whether this relation can be helpful with other kinds of constraints such as 2-connectivity and other generalizations of
\textsc{Feedback Vertex Set} such as the ones studied in~\cite{BonnetBKM19}.
	In particular, one could consider the following generalization of \textsc{Odd Cycle Transversal}:
	\pbDef{Subset Odd Cycle Transveral (SOCT)}{A graph $G$ and $S\subseteq V(G)$}{A set $X\subseteq V(G)$ of minimum weight such that $G[\comp{X}]$ does not contain an odd cycle that intersects $S$.}
	Similar to SFVS, we can solve SOCT in time $k^{O(k)}\cdot n^{O(1)}$ parameterized by treewidth and this is optimal under ETH~\cite{BergougnouxBBK20}.
	We do not know whether SOCT is in \XP parameterized by mim-width, though it is in \FPT parameterized by clique-width or rank-width, since we can express it in MSO$_1$ (with the characterization used in~\cite{BergougnouxBBK20}).

For many well-known graph classes a decomposition of bounded mim-width can be found in polynomial time. However, for general graphs it is known that computing mim-width is \W[1]-hard and not in \APX unless \NP=\ ZPP \cite{SaetherV16}, while Yamazaki \cite{Yamazaki18} shows that under the small set expansion hypothesis it is not in \APX unless \Poly=\ \NP.
For dynamic programming algorithms as in this paper, to circumvent the assumption that we are given a decomposition, we want functions $f, g$ and an algorithm that given a graph of mim-width OPT computes an $f(OPT)$-approximation to mim-width in time $n^{g(OPT)}$, i.e. \XP by the natural parameter. This is the main open problem in the field. The first task could be to decide if there is a constant $c$ and a polynomial-time algorithm that given a graph $G$ either decides that its mim-width is larger than 1 or else returns a decomposition of mim-width $c$.
	
	\bibliographystyle{abbrv}
	\bibliography{biblio}	

\begin{thebibliography}{10}

\bibitem{BelmonteV13}
R.~Belmonte and M.~Vatshelle.
\newblock Graph classes with structured neighborhoods and algorithmic
  applications.
\newblock {\em Theoret. Comput. Sci.}, 511:54--65, 2013.

\bibitem{BergougnouxBBK20}
B.~Bergougnoux, {\'{E}}.~Bonnet, N.~Brettell, and O.~Kwon.
\newblock Close relatives of feedback vertex set without single-exponential
  algorithms parameterized by treewidth.
\newblock In {\em 15th International Symposium on Parameterized and Exact
  Computation, {IPEC} 2020, December 14-18, 2020, Hong Kong, China (Virtual
  Conference)}, pages 3:1--3:17, 2020.

\bibitem{BergougnouxK19esa}
B.~Bergougnoux and M.~M. Kant{\'{e}}.
\newblock More applications of the d-neighbor equivalence: Connectivity and
  acyclicity constraints.
\newblock In {\em 27th Annual European Symposium on Algorithms, {ESA} 2019,
  September 9-11, 2019, Munich/Garching, Germany.}, pages 17:1--17:14, 2019.

\bibitem{Bodlaender06}
H.~L. Bodlaender.
\newblock Treewidth: Characterizations, applications, and computations.
\newblock In {\em Graph-Theoretic Concepts in Computer Science, 32nd
  International Workshop, {WG} 2006, Bergen, Norway, June 22-24, 2006, Revised
  Papers}, pages 1--14, 2006.

\bibitem{BodlaenderCKN15}
H.~L. Bodlaender, M.~Cygan, S.~Kratsch, and J.~Nederlof.
\newblock Deterministic single exponential time algorithms for connectivity
  problems parameterized by treewidth.
\newblock {\em Inform. and Comput.}, 243:86--111, 2015.

\bibitem{BonnetBKM19}
{\'{E}}.~Bonnet, N.~Brettell, O.~Kwon, and D.~Marx.
\newblock Generalized feedback vertex set problems on bounded-treewidth graphs:
  Chordality is the key to single-exponential parameterized algorithms.
\newblock {\em Algorithmica}, 81(10):3890--3935, 2019.

\bibitem{BuiXuanTV13}
B.-M. Bui-Xuan, J.~A. Telle, and M.~Vatshelle.
\newblock Fast dynamic programming for locally checkable vertex subset and
  vertex partitioning problems.
\newblock {\em Theoret. Comput. Sci.}, 511:66--76, 2013.

\bibitem{Calinescu08}
G.~Calinescu.
\newblock {\em Multiway Cut}.
\newblock Springer, 2008.

\bibitem{CLL09}
J.~Chen, Y.~Liu, and S.~Lu.
\newblock An improved parameterized algorithm for the minimum node multiway cut
  problem.
\newblock {\em Algorithmica}, 55:1--13, 2009.

\bibitem{ChitnisFLMRS17}
R.~H. Chitnis, F.~V. Fomin, D.~Lokshtanov, P.~Misra, M.~S. Ramanujan, and
  S.~Saurabh.
\newblock Faster exact algorithms for some terminal set problems.
\newblock {\em Journal of Computer and System Sciences}, 88:195--207, 2017.

\bibitem{fvs:chord:corneil:1988}
D.~G. Corneil and J.~Fonlupt.
\newblock The complexity of generalized clique covering.
\newblock {\em Discrete Applied Mathematics}, 22(2):109 -- 118, 1988.

\bibitem{CourcelleMR00}
B.~Courcelle, J.~A. Makowsky, and U.~Rotics.
\newblock Linear time solvable optimization problems on graphs of bounded
  clique-width.
\newblock {\em Theory Comput. Syst.}, 33(2):125--150, 2000.

\bibitem{CourcelleO00}
B.~Courcelle and S.~Olariu.
\newblock Upper bounds to the clique width of graphs.
\newblock {\em Discrete Applied Mathematics}, 101(1-3):77--114, 2000.

\bibitem{CPPW13}
M.~Cygan, M.~Pilipczuk, M.~Pilipczuk, and J.~O. Wojtaszczyk.
\newblock On multiway cut parameterized above lower bounds.
\newblock {\em {TOCT}}, 5(1):3:1--3:11, 2013.

\bibitem{CyganPPW13}
M.~Cygan, M.~Pilipczuk, M.~Pilipczuk, and J.~O. Wojtaszczyk.
\newblock Subset feedback vertex set is fixed-parameter tractable.
\newblock {\em {SIAM} J. Discrete Math.}, 27(1):290--309, 2013.

\bibitem{DahlhausJPSY94}
E.~Dahlhaus, D.~S. Johnson, C.~H. Papadimitriou, P.~D. Seymour, and
  M.~Yannakakis.
\newblock The complexity of multiterminal cuts.
\newblock {\em {SIAM} J. Comput.}, 23(4):864--894, 1994.

\bibitem{EvenNZ00}
G.~Even, J.~Naor, and L.~Zosin.
\newblock An 8-approximation algorithm for the subset feedback vertex set
  problem.
\newblock {\em SIAM J. Comput.}, 30(4):1231--1252, 2000.

\bibitem{FominGLS16}
F.~V. Fomin, S.~Gaspers, D.~Lokshtanov, and S.~Saurabh.
\newblock Exact algorithms via monotone local search.
\newblock In {\em Proceedings of STOC 2016}, pages 764--775, 2016.

\bibitem{FominGR18}
F.~V. Fomin, P.~A. Golovach, and J.~Raymond.
\newblock On the tractability of optimization problems on h-graphs.
\newblock In {\em 26th Annual European Symposium on Algorithms, {ESA} 2018,
  August 20-22, 2018, Helsinki, Finland}, pages 30:1--30:14, 2018.

\bibitem{FominHKPV14}
F.~V. Fomin, P.~Heggernes, D.~Kratsch, C.~Papadopoulos, and Y.~Villanger.
\newblock Enumerating minimal subset feedback vertex sets.
\newblock {\em Algorithmica}, 69(1):216--231, 2014.

\bibitem{GargVY04}
N.~Garg, V.~V. Vazirani, and M.~Yannakakis.
\newblock Multiway cuts in node weighted graphs.
\newblock {\em J. Algorithms}, 50(1):49--61, 2004.

\bibitem{GolovachHKKSV18}
P.~A. Golovach, P.~Heggernes, M.~M. Kant{\'{e}}, D.~Kratsch, S.~H. S{\ae}ther,
  and Y.~Villanger.
\newblock Output-polynomial enumeration on graphs of bounded (local) linear
  mim-width.
\newblock {\em Algorithmica}, 80(2):714--741, 2018.

\bibitem{GolovachHKS14}
P.~A. Golovach, P.~Heggernes, D.~Kratsch, and R.~Saei.
\newblock Subset feedback vertex sets in chordal graphs.
\newblock {\em J. Discrete Algorithms}, 26:7--15, 2014.

\bibitem{GolumbicR00}
M.~C. Golumbic and U.~Rotics.
\newblock On the clique-width of some perfect graph classes.
\newblock {\em Int. J. Found. Comput. Sci.}, 11(3):423--443, 2000.

\bibitem{HolsK18}
E.~C. Hols and S.~Kratsch.
\newblock A randomized polynomial kernel for subset feedback vertex set.
\newblock {\em Theory Comput. Syst.}, 62:54--65, 2018.

\bibitem{JacobBDP21}
H.~Jacob, T.~Bellitto, O.~Defrain, and M.~Pilipczuk.
\newblock Close relatives (of feedback vertex set), revisited.
\newblock {\em CoRR}, abs/2106.16015, 2021.

\bibitem{JaffkeKST19tcs}
L.~Jaffke, O.~Kwon, T.~J.~F. Str{\o}mme, and J.~A. Telle.
\newblock Mim-width {III.} {G}raph powers and generalized distance domination
  problems.
\newblock {\em Theor. Comput. Sci.}, 796:216--236, 2019.

\bibitem{JaffkeKT18}
L.~Jaffke, O.~Kwon, and J.~A. Telle.
\newblock A unified polynomial-time algorithm for feedback vertex set on graphs
  of bounded mim-width.
\newblock In {\em 35th Symposium on Theoretical Aspects of Computer Science,
  {STACS} 2018, February 28 to March 3, 2018, Caen, France}, pages 42:1--42:14,
  2018.

\bibitem{JaffkeKT20}
L.~Jaffke, O.~Kwon, and J.~A. Telle.
\newblock Mim-width {II.} the feedback vertex set problem.
\newblock {\em Algorithmica}, 82(1):118--145, 2020.

\bibitem{KK12}
K.~Kawarabayashi and Y.~Kobayashi.
\newblock Fixed-parameter tractability for the subset feedback set problem and
  the s-cycle packing problem.
\newblock {\em J. Comb. Theory, Ser. B}, 102(4):1020--1034, 2012.

\bibitem{Kim82}
K.~H. Kim.
\newblock {\em Boolean matrix theory and applications}, volume~70.
\newblock Dekker, 1982.

\bibitem{KratschMT08}
D.~Kratsch, H.~M{\"u}ller, and I.~Todinca.
\newblock Feedback vertex set on {AT}-free graphs.
\newblock {\em Discrete Applied Mathematics}, 156(10):1936--1947, 2008.

\bibitem{Oum05a}
S.-i. Oum.
\newblock Rank-width and vertex-minors.
\newblock {\em J. Combin. Theory Ser. B}, 95(1):79--100, 2005.

\bibitem{Oum09a}
S.-I. Oum.
\newblock Approximating rank-width and clique-width quickly.
\newblock {\em ACM Trans. Algorithms}, 5(1):Art. 10, 20, 2009.

\bibitem{OumSV13}
S.-i. Oum, S.~H. S{\ae}ther, and M.~Vatshelle.
\newblock Faster algorithms for vertex partitioning problems parameterized by
  clique-width.
\newblock {\em Theoret. Comput. Sci.}, 535:16--24, 2014.

\bibitem{OumS06}
S.-i. Oum and P.~Seymour.
\newblock Approximating clique-width and branch-width.
\newblock {\em J. Combin. Theory Ser. B}, 96(4):514--528, 2006.

\bibitem{PapadopoulosT19}
C.~Papadopoulos and S.~Tzimas.
\newblock Polynomial-time algorithms for the subset feedback vertex set problem
  on interval graphs and permutation graphs.
\newblock {\em Discrete Applied Mathematics}, 258:204--221, 2019.

\bibitem{PapadopoulosT20}
C.~Papadopoulos and S.~Tzimas.
\newblock Subset feedback vertex set on graphs of bounded independent set size.
\newblock {\em Theor. Comput. Sci.}, 814:177--188, 2020.

\bibitem{RobertsonS91}
N.~Robertson and P.~D. Seymour.
\newblock Graph minors. x. obstructions to tree-decomposition.
\newblock {\em J. Comb. Theory, Ser. {B}}, 52(2):153--190, 1991.

\bibitem{SaetherV16}
S.~H. S{\ae}ther and M.~Vatshelle.
\newblock Hardness of computing width parameters based on branch decompositions
  over the vertex set.
\newblock {\em Theor. Comput. Sci.}, 615:120--125, 2016.

\bibitem{Vatshelle12}
M.~Vatshelle.
\newblock {\em New width parameters of graphs}.
\newblock PhD thesis, University of Bergen, Bergen, Norway, 2012.

\bibitem{Yamazaki18}
K.~Yamazaki.
\newblock Inapproximability of rank, clique, boolean, and maximum induced
  matching-widths under small set expansion hypothesis.
\newblock {\em Algorithms}, 11(11):173, 2018.

\end{thebibliography}
	
	\appendix
	
	\newpage
	
	\section{Explanations with several figures}
	\label{appendix}
	
	The following figures explain the relation between solutions ($S$-forests) and an index $i$.

\begin{figure}[h!tb]
	\centering
	\includegraphics[width=0.9\linewidth,page=2]{explanations}

	\caption{Illustration of the $\comp{S}$-contraction $\cP_Y$ constructed in Lemma~\ref{lem:equiStreeScontrac}. The graph represented is an $S$-forest induced by the union of a set $X\subseteq V_x$ and a set $Y\subseteq \comp{V_x}$.
		The white filled disks represent the vertices in $S$.
		We transform this $S$-forest into a forest (see the next figure) by contracting the circled subsets of vertices.
		The $\comp{S}$-contraction used on $X$ is $\cc(X\setminus S)$.
		Observe that $Y_1$ is the only block of $\cP_Y$ which is not a connected component of $G[Y\setminus S]$. We need $Y_1$ to kill the orange cycle which is not an $S$-cycle.}
		\label{fig:killcycle}
\end{figure}

\begin{figure}[h!tb]
	\centering
	\includegraphics[width=0.9\linewidth,page=3]{explanations}
	\caption{The contracted graph $G[X\cup Y]_{\downarrow \cc(X\setminus S)\cup\cP_Y}$.
		The vertices of this graph are the blocks of $\cc(X\setminus S)\cup\cP_Y$ and the singletons $\{v\}$ for every vertex $v$ in $(X\cup Y)\cap S$, these singletons are represented by the white filled disks. }
		\label{fig:contracted}
\end{figure}

\begin{figure}[h!tb]
	\centering
	\includegraphics[width=\linewidth,page=4]{explanations}
	\caption{The bipartite graph $G[X,Y]_{\downarrow \cc(X\setminus S)\cup \cP_Y}$. The colored squares represent the blocks in a vertex cover $\VC$ of $G[X,Y]_{\downarrow \cc(X\setminus S)\cup \cP_Y}$ that satisfies the properties of Lemma~\ref{lem:equiStreeScontrac}.}
\end{figure}

\begin{figure}[htb]
	\centering
	\includegraphics[width=\linewidth,page=7]{explanations}
	\caption{The graph $G[X\cup Y]_{\downarrow \cc(X\setminus S)\cup\cP_Y}$. As done in Lemma~\ref{lemma:exitsAnIndex}, we construct from the vertex cover $\VC$ an index $i=(\cX_{\vc}^{\comp{S}},\cX_{\vc}^{S}, \cX_{\comp{\vc}},\cY_{\vc}^{\comp{S}},\cY_{\vc}^{S})\in\bI_x$ such that $X$ is a partial solution associated with $i$ and $(Y,\cP_Y)$ is a complement solution associated with $i$.\\
	The set $\cX_{\vc}^{\comp{S}}$ contains the representatives of the blocks in $\cc(X\setminus S)$ that are in the vertex cover $\VC$. In this example, we have $\cX_{\vc}^{\comp{S}}=\{ R_2 \}$ where $R_2=\rep{V_x}{2}(X_2)$.}
\end{figure}

\begin{figure}[htb]
	\centering
	\includegraphics[width=\linewidth,page=8]{explanations}
	\caption{The set $\cX_{\vc}^{S}$ contains the representatives of the singletons $\{x\}\in \VC$ with $x\in S$.
		In this example, we have $\cX_{\vc}^{S}=\{R_1\}$ where $R_1=\rep{V_x}{1}(\{x_1\})$.}
\end{figure}

\begin{figure}[htb]
	\centering
	\includegraphics[width=\linewidth,page=9]{explanations}
	\caption{The set $\cX_{\comp{\vc}}=\rep{V_x}{1}(X\setminus V(\VC))$.
		It is the representative set of the union of the blocks in $X_{\downarrow \cc(X\setminus S)} \setminus \VC$.}
\end{figure}

\begin{figure}[htb]
	\centering
	\includegraphics[width=\linewidth,page=10]{explanations}
	\caption{The set $\cY_{\vc}^{\comp{S}}$ contains the representatives of the blocks in $\cP_Y \cap \VC$.		
		In this example, we have $\cY_{\vc}^{\comp{S}}=\{\comp{R_1},\comp{R_2}\}$ where $\comp{R_\ell}=\rep{\comp{V_x}}{2}(Y_\ell)$ for $\ell=1,2$. }
\end{figure}

\begin{figure}[htb]
	\centering
	\includegraphics[width=\linewidth,page=11]{explanations}
	\caption{The set $\cY_{\vc}^{S}$ contains the representatives of the singletons $\{y\}\in\VC$ with $y\in S$.
		In this example, we have $\cY_{\vc}^{S}=\{\comp{R_3},\comp{R_4}\}$ where $\comp{R_\ell}=\rep{\comp{V_x}}{1}(\{y_\ell\})$ for $\ell=3,4$. }
\end{figure}

\begin{figure}[htb]
	\centering
	\includegraphics[width=\linewidth,page=12]{explanations}
	\caption{The auxiliary graph $\aux_x(X,i)$.
		It can be obtained from $G[X\cup Y]_{\downarrow\cc(X\setminus S)\cup \cP_Y}$ by (1) removing the edges between blocks of $Y_{\downarrow\cP_Y}$, (2) removing the blocks of $Y_{\downarrow\cP_Y}$ that do not belong to $\VC$ and (3) replacing each remaining blocks of $Y_{\downarrow\cP_Y}$ by its representatives.\\
		The blocks of the partition $\cc(X,i)$ are $\{R_1,\comp{R_1},\comp{R_2}\}$, $\{R_2,\comp{R_3}\}$ and $\{\comp{R_4}\}$ where $R_1$ and $R_2$ are the representatives of $\{x_1\}$ and $X_2$ respectively.}
	\label{fig:ccXi}
\end{figure}

\begin{figure}[htb]
	\centering
	\includegraphics[width=\linewidth,page=12]{explanations}\\
	\vspace{1cm}
	\includegraphics[width=\linewidth,page=14]{explanations}
	\caption{The auxiliary graphs $\aux_x(X,i)$ and $\aux_x(W,i)$ with $W$ a partial solution associated with $i$. 
	We have $R_1=\rep{V_x}{1}(\{w_1\})=\rep{V_x}{1}(\{x_1\})$ and $R_2=\rep{V_x}{2}(W_2)=\rep{V_x}{2}(X_2)$.
	Observe that $\cc(X,i)=\cc(W,i)=\{\{R_1,\comp{R_1},\comp{R_2}\},\{R_2,\comp{R_3}\},\{\comp{R_4}\}\}$. Consequently, $X$ and $W$ are $i$-equivalent.}
	\label{fig:equivalent}
\end{figure}

\begin{figure}[htb]
	\centering
	\includegraphics[width=\linewidth,page=5]{explanations}\\
	\vspace{1cm}
	\includegraphics[width=\linewidth,page=15]{explanations}
	\caption{The graphs $G[X\cup Y]_{\downarrow \cc(X\setminus S)\cup\cP_Y}$ and $G[W\cup Y]_{\downarrow \cc(W\setminus S)\cup\cP_Y}$.
	As $X$ and $W$ are $i$-equivalent and $G[X\cup Y]_{\downarrow \cc(X\setminus S)\cup\cP_Y}$ is a forest, by Lemma~\ref{lemma:equivalence},  $G[W\cup Y]_{\downarrow \cc(W\setminus S)\cup\cP_Y}$ is also a forest. Thus, by Lemma~\ref{lemma:forestImpliesS-forest}, $G[W\cup Y]$ is an $S$-forest. In fact, these two lemmas implies that for every complement solution $(Y',\cP_Y')$ associated with $i$, $G[X\cup Y']$ is an $S$-forest if and only if $G[W\cup Y']$ is an $S$-forest.}
\end{figure}
	
\begin{figure}[htb]
	\centering
	\includegraphics[width=\linewidth,page=6]{explanations}\\
	\vspace{1cm}
	\includegraphics[width=\linewidth,page=16]{explanations}
	\caption{New examples for the graphs $G[X\cup Y]_{\downarrow \cc(X\setminus S)\cup\cP_Y}$ and $G[W\cup Y]_{\downarrow \cc(W\setminus S)\cup\cP_Y}$.\\
	To prove Lemma~\ref{lemma:equivalence}, we prove that if  $G[W\cup Y]_{\downarrow \cc(W\setminus S)\cup\cP_Y}$ contains a cycle, then  $G[X\cup Y]_{\downarrow \cc(X\setminus S)\cup\cP_Y}$ contains a cycle. We use the following arguments.\\		
	The orange paths exist in both graphs because these paths only use blocks in $Y_{\downarrow \cP_Y}$.
	Since $\cc(X,i)=\cc(W,i)$, the purple path in $G[W\cup Y]_{\downarrow \cc(W\setminus S)\cup\cP_Y}$ implies the existence of the purple path in $G[X\cup Y]_{\downarrow \cc(X\setminus S)\cup\cP_Y}$.
	Finally, as $X_2\equiv^{V_x}_{2} W_2$, the blocks in $Y_{\downarrow \cP_Y}$  adjacent to $W_2$ are also adjacent to $X_2$. Thus, the green path in in $G[W\cup Y]_{\downarrow \cc(W\setminus S)\cup\cP_Y}$ implies the existence of the green path in $G[X\cup Y]_{\downarrow \cc(X\setminus S)\cup\cP_Y}$.}
	\label{fig:intuitionCycles}
\end{figure}

\end{document}